%% file: ms.tex
\pdfoutput=1	

\documentclass[11pt]{article}
\usepackage{graphicx}
\usepackage{fullpage}

\usepackage{xspace}
\usepackage{subcaption}
\usepackage{color}

\usepackage{enumitem}

\usepackage[linesnumbered, ruled, vlined]{algorithm2e}

\usepackage{quotes}

\usepackage{amsmath,amsthm,amssymb}

\newcommand{\stream}{\mathcal{R}}

\newtheorem{theorem}{Theorem}
\newtheorem{definition}{Definition}
\newtheorem{lemma}{Lemma}

\newtheorem{observation}{Observation}
\newcommand{\remove}[1]{}

\newcommand{\eps}{\varepsilon}

\newcommand{\graphwidth}{0.8\columnwidth}

\newcommand{\ssunif}{\texttt{Uniform}\xspace}		
\newcommand{\voila}{{\tt VOILA}\xspace}
\newcommand{\svoila}{\texttt{S-VOILA}\xspace}

\newcommand{\reservoir}{\texttt{Reservoir}\xspace}	
\newcommand{\neyman}{\texttt{Neyman}\xspace}		
\newcommand{\neymanPlus}{\texttt{Neyman+}\xspace}	
\newcommand{\exact}{\texttt{Exact}\xspace}

\newcommand{\eg}{\hbox{\emph{e.g.,}}\xspace}
\newcommand{\ie}{\hbox{\emph{i.e.,}}\xspace}
\newcommand{\etal}{\hbox{\emph{et al.}}\xspace}

\newcommand{\srs}{{stratified random sample}}

\newcommand{\sssr}{{\tt SingleSSR}} 
                             
\newcommand{\ssr}{{\tt SSR}}
\newcommand{\fssr}{{\tt FastSSR}}

\newcommand{\vob}{{\tt VOILA}\xspace}

\DeclareMathOperator*{\argmax}{arg\,max}
\DeclareMathOperator*{\argmin}{arg\,min}

\SetKwProg{Fn}{Function}{}{}

\title{Variance-Optimal Offline and Streaming\\Stratified  Random  Sampling}
\author{Trong Duc Nguyen$^1$ \and Ming-Hung Shih$^1$ \and Divesh
  Srivastava$^2$ \and Srikanta Tirthapura$^1$ \and Bojian Xu$^3$}
\date{%
    $^1$Iowa State
    University, IA 50011, USA\\%
    $^2$AT\&T Labs--Research, NJ 07921, USA\\%
    $^3$ Eastern Washington
    University, WA 99004, USA.\\
\ \\
   \texttt{trong@iastate.edu\ \  mshih@iastate.edu\ \ 
     divesh@research.att.com\ \ snt@iastate.edu\ \
     bojianxu@ewu.edu}\\
\ \\
    \today
}

\begin{document}

\maketitle

\begin{abstract}
Stratified random sampling (SRS) is a fundamental sampling technique that provides accurate estimates for aggregate queries using a small size sample, and has been used widely for approximate query processing. A key question in SRS is how to partition a target sample size among different strata. While {\em Neyman allocation} provides a solution that minimizes the variance of an estimate using this sample, it works under the assumption that each stratum is abundant, \ie has a large number of data points to choose from. This assumption may not hold in general: one or more strata may be bounded, and may not contain a large number of data points, even though the total data size may be large. 

We first present \voila, an offline method for allocating sample sizes to strata in a variance-optimal manner, even for the case when one or more strata may be bounded. We next consider SRS on streaming data that are continuously arriving. We show a lower bound, that any streaming algorithm for SRS must have (in the worst case) a variance that is $\Omega(r)$ factor away from the optimal, where $r$ is the number of strata. We present \svoila, a practical streaming algorithm for SRS that is {\em locally variance-optimal} in its allocation of sample sizes to different strata. Our result from experiments on real and synthetic data show that \voila can have significantly (1.4 to 50.0 times) smaller variance than Neyman allocation. The streaming algorithm \svoila results in a variance that is typically close to \voila, which was given the entire input beforehand. 
\end{abstract}



\input{intro}
\input{related}

\input{prelims}

\input{reduction}

\input{offline}

\input{online}

\input{experiment}
\input{conclusion}
\clearpage

\bibliographystyle{abbrv}
\bibliography{streams}




\end{document}

%% file: intro.tex
\section{Introduction}
\label{sec:intro}
Random sampling is a widely-used method for data analysis, and features prominently in the toolbox of virtually every approximate query processing system. The power of random sampling lies in its generality. For many important classes of queries, an approximate answer, whose error is small in a statistical sense, can be efficiently obtained through executing the query over an appropriately derived random sample. Sampling operators are part of all major database products, \eg Oracle, Microsoft SQL Server, and IBM Db2.

The simplest method for random sampling is {\em uniform random sampling}, where each element from the entire data (the "population") is chosen with the same probability.  But uniform random sampling may lead to a high variance in estimates for aggregate queries. For instance, consider a population $D = \{1,1000,2,4,2,1050,1200,1,1300\}$, and suppose we wanted to estimate the sum of the population. A uniform random sample of size two will lead to an estimate with a variance of $1.3\times10^7$. 

An alternative sampling method is {\em stratified random sampling} (SRS), where the population is partitioned into subgroups called "strata". Within each stratum, uniform random sampling is used to select a per-stratum sample. The different per-stratum samples are then combined to derive the "stratified random sample". Suppose that the population is divided into two strata, one with elements $\{1,2,4,2,1\}$ and the other with elements $\{1000,1050,1200,1300\}$. A stratified random sample of size two can choose one element from each stratum, yielding an estimate
with the variance of $2\times10^5$, 46 times smaller than what was possible with a uniform random sample of the same size.

In SRS, there is flexibility to emphasize some strata over others, through controlling the allocation of sample sizes; for instance, a stratum with a high standard deviation of values within can be given a larger allocation than another stratum with a lower standard deviation. In the above example, if we desire a stratified sample of size three, it is beneficial to allocate a smaller sample size of one to the first stratum and a larger sample size of two for the second stratum, since the standard deviation of the second stratum is higher. Doing so, the variance of estimate of the population sum reduces to approximately $1 \times 10^5$. SRS has been used widely in database systems for approximate query processing~\cite{AMPMMS-eurosys2013,CDN-TODS2007, JJ-ICDE2008, AGPR-SIGMOD1999, AGP-SIGMOD2000, BCD-SIGMOD2003}. 


Suppose that there are $r$ strata, numbered from $1$ to $r$, and that the mean, standard deviation, and number of items in the $j$th stratum are $\mu_j$, $\sigma_j$, and $n_j$ respectively.  Suppose that the target sample size is $M$ (total across all strata). We measure the quality of a stratified random sample through the variance in the estimate of the population mean\footnote{The standard deviation of data within a stratum is distinct from the variance of an estimate of an aggregate that is derived from a stratified random sample.}, computed using this sample. In "uniform allocation", each stratum $j$ gets an identical allocation of sample size of $s_j = M/r$. In "proportional allocation", a stratum is allocated a sample size proportional to the number of elements in it. A commonly used method that is believed to yield the smallest variance for an estimate of a population mean is "Neyman allocation"~\cite{Neyman1934,Coch77:book}, where stratum $j$ gets an allocation proportional to $\sigma_j n_j$. Many sampling methods for approximate query processing, such as the ones used in~\cite{CDN-TODS2007,AMPMMS-eurosys2013}, are based on Neyman allocation.


A problem with Neyman allocation is that it assumes that each stratum has abundant data, much larger than the size of samples. However, in practice, strata can be bounded, and may not always contain a large number of elements, and in such situations, Neyman allocation can be suboptimal. To see this, suppose there were 10 strata in the population, and suppose stratum 1 had 100 items and a standard deviation of 100, strata 2 to 10 each had 1000 items and a standard deviation of 0.1. 
With a sample size of $M=1000$ items ($\approx 11 \%$ of data size), Neyman allocation assigns 917 samples to stratum 1, and 9 samples each to the other strata. However stratum 1 has only 100 items, and it is wasteful to allocate more samples to this stratum. We call such strata, which have a small number of elements relative to the assigned sample size, as "bounded" strata. For instance, in our experiments with a sample size of 1 million from the one-year-long OpenAQ dataset~\cite{OpenAQ} on air quality measurements, we found that after the first month, 11 out of 60 strata are bounded.  For data with bounded strata, Neyman allocation is clearly no longer the variance-optimal method for sample size allocation.


Another problem with the current state-of-the-art is that methods for SRS are predominantly offline methods, and assume that all data is available before sampling starts. As a result, systems that rely on SRS (\eg~\cite{AMPMMS-eurosys2013,CDN-TODS2007}) cannot easily adapt to new data arrivals and will need to recompute stratified random samples from scratch, as more data arrives. However, with the advent of streaming data warehouses such as Tidalrace~\cite{JS-CIDR2015}, it is imperative to have methods for SRS that work on dynamic data streams, and maintain stratified random samples in an incremental manner. In this work, we consider the general problem of variance-optimal SRS in both the offline and streaming settings, when some of the strata may be bounded.

\subsection{Our Contributions}
\label{sec:cont}
\noindent {\bf -- Variance Optimal Stratified Random Sampling:} We present the first offline algorithm for variance-optimal SRS for data that may have bounded strata. Our algorithm \voila ({\tt V}ariance {\tt O}pt{\tt I}ma{\tt L} {\tt A}llocation) computes an allocation that has provably optimal variance among all possible allocations of sample sizes to different strata. While prior work assumes that there are no strata with small volumes of data, which is often violated in real data sets, our analysis makes no such assumptions. \voila is a generalization of Neyman allocation and reduces to Neyman allocation in the case when every stratum is abundant.

\smallskip

\noindent {\bf -- Streaming Lower Bound:} We present a lower bound showing that any streaming algorithm for SRS that uses a memory of $M$ records must have, in the worst case, a variance that is a factor of $\Omega(r)$ away from the variance of the optimal offline algorithm for SRS that uses a memory of $M$ records. This lower bound is tight, since there exist streaming algorithms for SRS whose variance matches this bound in the worst case.

\smallskip

\noindent {\bf -- Practical Streaming Algorithm for SRS:}
We present \svoila, a streaming algorithm for SRS that is locally optimal with respect to variance -- upon receiving new elements, it (re-)allocates sample sizes among strata so as to minimize the variance among all possible re-allocations. \svoila can be viewed as the online, or dynamic counterpart of the optimization that led to \voila, which is based on optimizing the variance using a static view of data. \svoila can also deal with the case when a minibatch of multiple data items is seen at a time, rather than only a single item at a time -- re-allocations made by \svoila are locally optimal with respect to the entire minibatch, and are of higher quality for larger size minibatches than when a single element is seen at a time. In our experimental study, we found that the variance of \svoila is typically close to that of the offline algorithm \voila, and the variance of \svoila improves as the size of the minibatch increases. Since it can deal with minibatches of varying sizes, it is well-suited to real-world streams that may have bursty arrivals.

\smallskip

\noindent {\bf -- Variance Optimal Sample Size Reduction:}
The algorithms for offline SRS (\voila) and streaming SRS (\svoila) are both based on a technique for reducing the size of an existing stratified random sample down to a desired target size such that the increase in variance of the estimator based on the final sample is optimized. This technique for sample size reduction may be of independent interest in other tasks such as sub-sampling from a given stratified random sample. 

\smallskip

\noindent {\bf -- Experimental Evaluation:}
We present a detailed experimental evaluation using real and synthetic data, considering both quality of sample and accuracy of query answers. Our experiments show that (a)~\voila  can have significantly smaller variance than Neyman allocation, and (b) \svoila closely tracks the allocation as well as the variance of the optimal offline algorithm \voila. As the size of the minibatch increases, the variance of the samples produced by \svoila decreases. A minibatch of size 100 provides most of the benefits of \voila, in our experiments on real-world data.

%% file: related.tex
\subsection{Related Work}
\label{sec:related}
Sampling has been widely used in approximate query processing on both static and streaming data~\cite{Coch77:book,Lohr-book2009,YTill-book1997,SKT2012,HAAS2016}. The reservoir sampling~\cite{MB1983,Vitter-sampling-focs83} method for maintaining a uniform random sample on a stream has been known for decades, and many variants have been considered, such as weight-based sampling~\cite{ES-IPL2006,BOV-sampling2015}, stream sampling under insertion and deletion of elements~\cite{GLH-vldbj2008}, distinct sampling~\cite{GT01}, sampling from a sliding window~\cite{BDM02,GL-SIGMOD2008,BOZ-PODS2009}, and time-decayed sampling~\cite{CTX-SICOMP09,CSSX09}.

SRS in the online setting~\cite{SKT2012} can be viewed as a type of weight based reservoir sampling where the weight of each stream element is changing dynamically, based on the statistics of the stratum the element belongs to. Since the weight of each stream element changes dynamically, even after it has been observed, prior work on weighted reservoir sampling~\cite{ES-IPL2006} does not apply here, since it assumes that the weight of an element is known at the time of observation and does not change henceforth. Meng \cite{Meng-ML2013} considered streaming SRS using population-based allocation. Al-Kateb~\etal~\cite{AL-srs-SSDBM2010,AL-IS2014} considered streaming SRS using power allocation, based on their prior work on adaptive reservoir sampling~\cite{ALW-SSDBM2007}. Lang \textit{et al.}~\cite{LLS15} consider machine learning methods for determining the per-item probability of inclusion in a sample. This work is meant for static data, and can be viewed as a version of weighted random sampling where the weights are learnt using a query workload. Prior work on streaming SRS neither considers provable guarantees on the quality of the resulting samples, nor lower bounds for streaming SRS, like we do here.

A majority of prior work on using SRS in approximate query processing~\cite{AGPR-SIGMOD1999, AGP-SIGMOD2000, BCD-SIGMOD2003, CDN-TODS2007, JJ-ICDE2008,AMPMMS-eurosys2013} has assumed static data. With the emergence of data stream processing systems~\cite{BBDMW02} and data stream warehousing systems~\cite{JS-CIDR2015}, it is important to devise methods for streaming SRS with quality guarantees.

%% file: prelims.tex
\vspace*{-0.2cm}
\section{Overview}
\label{sec:overview}
\vspace*{-0.1cm}
\subsection{Preliminaries}
We consider the construction and maintenance of a stratified random sample of data that is either stored offline, or arriving as a stream. Stratified sampling can be viewed as being composed of three parts -- stratification, sample allocation, and sampling. 

Stratification is a partitioning of the universe into a number of disjoint strata, such that the union of all strata equals the universe. Equivalently, it is the assignment of each data element to a unique stratum. Stratification is often a pre-defined function of one or more attributes of the data element. For example, the work of Chaudhuri et al.~\cite{CDN-TODS2007} stratifies tuples within a database table based on the set of selection predicates in the query workload that the tuple satisfies. In the OpenAQ dataset~\cite{OpenAQ}, air quality data measurements can be stratified on the basis of geographic location and measurement type, so that tuples relevant to a query can typically be composed of the union of strata. Our work assumes that the universe has already been partitioned into strata, and that each tuple comes with a stratum identifier. This assumption fits the model assumed in~\cite{CDN-TODS2007,AMPMMS-eurosys2013}. 

Our work deals with sample allocation, the task of partitioning the available memory budget of $M$ samples among the strata. In the case of offline sampling, allocation needs to be done only once, after knowing the data in its entirety. In the case of streaming sampling, the allocation may need to be continuously re-adjusted as more data arrives, and the characteristics of different strata change. 

The final sampling step chooses within each stratum, the assigned number of samples uniformly at random. In the case of offline stratified sampling, the sampling step can be performed in a second pass through the data after sample size allocation, using reservoir sampling on the subset of elements belonging to each stratum. In the case of streaming sampling, the sampling step is not as easy, since it needs to occur simultaneously with sample (re-)allocation, which may change the allocations to different strata over time.

\smallskip
{\bf Variance-Optimal Allocation.}
Given a data set, $R=\{v_1, v_2, \ldots, v_n\}$ of size $n$, whose elements are stratified into $r$ strata, numbered $1,2,\ldots,r$.  For each $i=1\ldots r$, let $S_i$ be a a uniform random sample of size $s_i$ drawn without replacement from stratum $i$. Let $\mathbb{S} = \{S_1, S_2, \ldots, S_n\}$ denote the stratified random sample.  

The sample mean of each per-stratum sample $S_i$ of size $s_i$ is: $\bar{y}_i = \frac{\sum_{v \in S_i} v }{s_i}$.
The population mean of $R$, $\mu_R$ can be estimated as:
$\bar{y} = \frac{\sum_{i=1}^r n_i \bar{y}_i}{n}$, using the sample means of all strata. It can be shown that the expectation of $\bar{y}$ equals $\mu_R$.

Given a memory budget of $M \leq n$ elements to store all the samples, so that $\sum_i s_i = M$, we address the question: What is the value of each $s_i$, the size of sample $S_i$, so as to minimize the variance of $\bar{y}$. The variance of $\bar{y}$ can be computed as follows (e.g. see Theorem~5.3 in~\cite{Coch77:book}):
\begin{eqnarray}
V = V(\bar{y}) =  \frac{1}{n^2} \sum_{i=1}^r n_i(n_i -
s_i)\frac{\sigma_i^2}{s_i} 
= \frac{1}{n^2} \sum_{i=1}^r \frac{n_i^2 \sigma_i^2}{s_i} - 
\frac{1}{n^2} \sum_{i=1}^r n_i\sigma_i^2. \label{eqn:var}
\end{eqnarray}
We call the answer to this question as a \emph{variance-optimal allocation} of sample sizes to different strata.

\smallskip
{\bf Neyman Allocation for Strata that are abundant.} All previous studies on variance-optimal  allocation assume that every stratum has a large volume of data, to fill its sample allocation. Under this assumption, Neyman allocation~\cite{Neyman1934,Coch77:book} minimizes the variance $V$, and allocates a sample size for stratum $i$  as  
$M\cdot (n_i \sigma_i) /\left(\sum_{j=1}^r n_j\sigma_j\right)$.

Given a collection of data elements $R$, we say a stratum $i$ is \emph{abundant}, if $n_i \geq M\cdot (n_i \sigma_i) /\left(\sum_{j=1}^r n_j\sigma_j\right)$. Otherwise, the stratum $i$ is \emph{bounded}. Clearly, Neyman allocation is optimal only if each stratum is abundant. It no longer be optimal if one or more strata are bounded. We consider the general case of variance-optimal allocation where there may be bounded strata.


\subsection{Solution Overview}
We note that both offline and streaming SRS can be viewed as a problem of "sample size reduction" in a variance-optimal manner. With offline SRS, we can initially view the entire data as a (trivial) sample of zero variance, where the sample size is very large -- this sample needs to be reduced to fit within the memory budget of $M$ records. 
If this reduction is done in a manner that minimizes the increase of variance, 
the resulting sample is a variance-optimal sample of size $M$. 

In the case of streaming SRS, the streaming algorithm maintains a current stratified random sample of size $M$. It also maintains the characteristics of each stratum, including the number of elements $n_i$ and standard deviation $\sigma_i$, in a streaming manner using $O(1)$ space per stratum. When a set of new stream elements arrive, we can let the per-stratum reservoir sampling algorithms continue sampling as before. If the sample size increases due to this step, then we are again faced with a problem of sample size reduction -- how can this be reduced to a sample of size $M$ in a variance-optimal manner? 

Based on the above observation, we first present a variance-optimal sample size reduction method in Section~\ref{sec:reduction}. We start with an algorithm for reducing the size of the sample by one element, followed by a general algorithm for reducing the size by $\beta \ge 1$ elements, and then present an improved algorithm with a faster runtime. The variance-optimal offline algorithm \voila can be viewed as an application of sample size reduction -- details are presented in Section~\ref{sec:offline}. We present a tight lower bound for any streaming algorithm, followed by \svoila, an algorithm for streaming SRS in Section~\ref{sec:streaming}. Note that the streaming algorithm \svoila does not necessarily lead to a variance-optimal sample. Though the individual sample-size reduction steps performed during observation of the stream are locally optimal, the overall result may not be optimal. Further details are in Section~\ref{sec:streaming}. We present a detailed experimental study of our algorithms in Section~\ref{sec:eval}.

%% file: reduction.tex
\section{Variance-Optimal Sample Size Reduction}
\label{sec:reduction}
\newcommand{\strata}{{\cal A}}
\newcommand{\vor}{\texttt{VOR\xspace}}
Suppose it is necessary to reduce an SRS of total size $M$ to an SRS of total size $M' < M$. This will need to reduce the size of the samples of one or more strata in the SRS. Since the sample sizes are reduced, the variance of the resulting estimate will increase. We consider the task of {\em variance-optimal sample size reduction (\vor)}, \ie how to partition the reduction in sample size among the different strata in such a way that the increase in the variance is minimized.



Consider Equation~\ref{eqn:var} for the variance of an estimate derived from the stratified random sample. Note that, for a given data set, a change in the sample sizes of different strata $s_i$ does not affect the parameters $n$, $n_i$, and $\sigma_i$. \vor~ can be formulated as the following non-linear program.
\begin{eqnarray}
\text{Minimize } \sum_{i=1}^r \frac{n_i^2\sigma_i^2}{s'_i}\label{eqn:obj1}
\end{eqnarray}
\noindent
subject to constraints: 
\begin{eqnarray}
0\leq s'_i \leq s_i \textrm{ for each }i=1,2,\ldots, r\label{eqn:cons1.1}\\
\sum_{i=1}^r s'_i = M'\label{eqn:cons1.2}
\end{eqnarray}

We observe that, without Constraint~\ref{eqn:cons1.1}, and if all strata are unbounded, the answer to the above optimization program is exactly the Neyman allocation under memory budget $M'$.  However, we have to deal with the additional Constraint~\ref{eqn:cons1.1} and the possibility of a stratum being bounded, in an efficient manner. In the rest of this section, we present efficient approaches for computing the \vor. 
\subsection{Special Case: Reduction by One Element}
\label{sec:reduction-single}
We first present an efficient algorithm for the case where the size of a \srs\ is reduced by one element.  An example application of this case is in designing a streaming algorithm for SRS, when stream items arrive one at a time.

We introduce a terminology that we will use frequently in the rest of the paper. 
Given a memory budget $M$, the Neyman allocation size for stratum $i$ is $M_i = M \cdot n_i\sigma_i/\sum_{j=1}^{r}n_j\sigma_j$. 
The task is to eliminate a random element from a stratum $i$ such that after reducing the sample size $s_i$ by one, the increase in variance $V$ (Equation~\ref{eqn:var}) is the smallest. 
Our solution is to choose stratum $i$ such that the partial derivative of $V$ with respect to $s_i$ is the largest over all possible choices of $i$. 
\[
\frac{\partial V}{\partial s_i} = -\frac{n_i^2 \sigma_i^2}{n^2} \frac{1}{s_i^2}.
\]
We choose stratum $\ell$ where:
\begin{eqnarray}
\ell = \argmax_i\left\{ \frac{\partial V}{\partial s_i} \,\middle\vert\,   1\leq i\leq r\right\}
= \argmin_i\left\{ \frac{n_i\sigma_i}{s_i}\,\middle\vert\, 1\leq i\leq r\right\}
= \argmax_i\left\{ \frac{s_i}{M'_i}\,\middle\vert\,
  1\leq i\leq r\right\}\label{eqn:per-ev},
\end{eqnarray}
where $M'_i$ is the Neyman allocation size for stratum $i$ under the new memory budget $M'$. Equation~\ref{eqn:per-ev} is due to the fact that each $M'_i$ is proportional to $n_i\sigma_i$. This gives the following lemma.

\begin{lemma}
\label{lem:one-ev}
When required to reduce the size of an \srs\ by one, the increase in variance of the estimated population mean is minimized if we reduce the size of $S_\ell$ by one, where
$\ell = \argmin_i\left\{ \frac{n_i\sigma_i}{s_i}\,\middle\vert\, 1 \leq i \leq r \right\}$.
\end{lemma}

In the case where we have multiple choices for $\ell$ using Lemma~\ref{lem:one-ev}, we choose the one where the current sample size $s_\ell$ is the largest. 
Algorithm \sssr~ for reducing the sample size by one is shown in Algorithm~\ref{algo:reduction-single}. It is straightforward to observe that the run time of the algorithm is $O(r)$.



%




\begin{algorithm}[t]
\DontPrintSemicolon
\caption{$\sssr()$: Variance-Optimal Sample  Size Reduction by One}
\label{algo:reduction-single}

\Return{$\argmin_i\left\{ \frac{n_i\sigma_i}{s_i}\,\middle\vert\, 1 \leq i \leq
  r \right\}$}\tcc*{The id of the stratum whose sample size shall be reduced by one.} 
\end{algorithm}

\subsection{Reduction by {\large $\beta\geq 1$} Elements}
\label{sec:reduction-general}
We now consider the general case, where the sample size needs to be reduced by some number $\beta$,  $1\leq \beta \leq M$. A possible solution idea is to repeatedly apply the one-element reduction algorithm (Algorithm~\ref{algo:reduction-single} from Section~\ref{sec:reduction-single}) $\beta$ times. Each iteration, a single element is chosen from a stratum such that the overall variance increases by the smallest amount. However, this greedy approach may not yield a sample with the smallest resulting variance. On the other hand, an exhaustive search of all possible evictions is not feasible either, since the number of possible ways to partition a reduction of size $\beta$ among $r$ strata is ${\beta+r-1} \choose {r}$, which is exponential in $r$ and a high degree polynomial in $\beta$, which can be very large. We now present efficient approaches to $\vor$. We first present a recursive algorithm, followed by a faster iterative algorithm. 
Before presenting the algorithm, we present the following useful characterization of a variance-optimal reduction.

\begin{definition}
\label{def:oversized}
We say that stratum $i$ is \emph{oversized} under memory budget $M$,  if its allocated sample size $s_i >  M_i$. Otherwise, we say that stratum $i$ is \emph{not oversized}.
\end{definition}
\begin{lemma}
\label{lem:evict-o}
Suppose that $E$ is the set of $\beta$ elements that are to be evicted from a stratified random sample such that the variance $V$ after eviction is the smallest possible. Then, each element in $E$ must be from a stratum whose current sample size is oversized under the new memory budget $M'=M - \beta$.
\end{lemma}
\vspace*{-0.3cm}
\begin{proof}
  We use proof by contradiction. Suppose one of the evicted elements,
  is deleted from a sample $S_\alpha$ such that the sample size $s_\alpha$ is not oversized under the new memory budget. 
  Because the order of the eviction of the $\beta$ elements
  does not impact the final variance, suppose that element $e$ is
  evicted after the other $\beta-1$ evictions have happened. Let
  $s_\alpha$ denote the size of sample $S_\alpha$ at the moment $t$
  right after the first $\beta-1$ evictions and before evicting
  $e$. The increase in variance caused by evicting an element from $S_\alpha$ is
\begin{eqnarray*}
\Delta = \frac{1}{n^2}\left(\frac{n_\alpha^2\sigma_\alpha^2}{s_\alpha(s_\alpha-1)}\right)
= \left(\frac{\sum_{i=1}^rn_i\sigma_i}{nM'}\right)^2\frac{{M'}_\alpha^2}{s_\alpha(s_\alpha-1)}
> \left(\frac{\sum_{i=1}^rn_i\sigma_i}{nM'}\right)^2
\end{eqnarray*}
where ${M'}_\alpha =M' \frac{n_\alpha\sigma_\alpha}{\sum_{i=1}^rn_i\sigma_i}$
is the Neyman allocation for stratum $\alpha$ under memory budget $M'$. The last
inequality is due to the fact that $S_\alpha$ is not oversized under
budget $M'$ at time $t$, i.e., $s_\alpha \leq M'_\alpha$.

Note that an oversized sample exist at time $t$, since there are
a total of $M'+1$ elements in the \srs\ at time $t$, and the memory
target is $M'$. Instead of evicting $e$, if we choose to evict another
element $e'$ from an oversized sample $S_{\alpha'}$, the resulting
increase in variance will be:
\begin{eqnarray*}
\Delta' = \frac{1}{n^2}\left(\frac{n_{\alpha'}^2\sigma_{\alpha'}^2}{s_{\alpha'}(s_{\alpha'}-1)}\right)
= \left(\frac{\sum_{i=1}^rn_i\sigma_i}{nM'}\right)^2\frac{{M'}_{\alpha'}^2}{s_{\alpha'}(s_{\alpha'}-1)}
< \left(\frac{\sum_{i=1}^rn_i\sigma_i}{nM'}\right)^2
\end{eqnarray*}
where $M'_{\alpha'} =M'\frac{n_{\alpha'}\sigma_{\alpha'}}{\sum_{i=1}^rn_i\sigma_i}$
is the Neyman allocation for stratum ${\alpha'}$ under
memory budget $M'$. The last inequality is due to the fact that
$S_{\alpha'}$ is oversized under budget $M'$ at time $t$, i.e.,
$s_{\alpha'} > M'_{\alpha'}$.  Because $\Delta' < \Delta$, at time
$t$, evicting $e'$ from $S_{\alpha'}$ leads to a lower variance than
evicting $e$ from $S_{\alpha}$. This is a contradiction to the
assumption that evicting $e$ leads to the smallest variance, and completes the proof.
\end{proof}

Lemma~\ref{lem:evict-o} implies that it is only necessary to 
reduce the size of the samples that are oversized under the target
memory budget $M'$. Samples that are not
oversized can be given their current allocation, even under the new
memory target $M'$. Our algorithm based on this observation first
allocates sizes to the samples that are not oversized. The remaining
memory now needs to be allocated among the oversized samples. Since
this can again be viewed as a sample size reduction problem, while
focusing on a smaller set of (oversized) samples, this is accomplished
using a recursive call under a reduced memory budget; See
Lemma~\ref{lem:recur} for a formal statement of this idea. The base
case for this recursion is when all samples under consideration are
oversized. In this case, we simply use the Neyman allocation to each
stratum, under the reduced memory budget $M'$
(Observation~\ref{ob:exit}). Our algorithm $\ssr$ is shown in
Algorithm~\ref{algo:reduction-batch}.

\begin{algorithm}[t]
\DontPrintSemicolon
\caption{\ssr($\strata, M, \mathcal{L}$): Variance-Optimal Sample Size Reduction}
\label{algo:reduction-batch}
\KwIn{$\strata$ -- set of strata under consideration. \\\hspace*{10mm} $M$ -- target sample size for all strata in $\strata$.}

\KwOut{For $i \in \strata$, $\mathcal{L}[i]$ is the final size of sample for stratum $i$.}

 $\mathcal{O} \gets \emptyset$~~ \tcp{oversized samples} 

 \For{$j \in \strata$\label{line:batch-neyman-2.1}}
 {
   $M_{j} \leftarrow M\cdot n_{j}\sigma_{j}/\sum_{t \in \strata} n_{t}\sigma_{t}$ \label{line:batch-neyman-2.2}
   \tcp{Neyman allocation if memory $M$ divided among $\strata$}
   \lIf{($s_{j} > M_{j}$)} {$\mathcal{O} \gets \mathcal{O} \cup \{j \}$}
   \lElse{$\mathcal{L}[j] \gets s_j$    \tcp{Keep current allocation}}
 }
  \If{$\mathcal{O} = \strata$ \label{line:evict-rec-exit-1}}{
   \tcp{All samples oversized. Recursion stops.}
   \lFor{$j\in \strata$}
    { 
       $\mathcal{L}[j] \gets M_j$
    } 
 }

\Else{

  \tcp{Recurse on $\mathcal{O}$, under remaining mem budget.}
   \ssr($\mathcal{O}, M-\sum_{j\in \strata - \mathcal{O}}s_j, \mathcal{L}$)\label{line:evict-rec}\;
 }
\end{algorithm}

Let $\mathbb{S} = \{S_1, S_2, \ldots, S_r\}$ be the current stratified random sample.  Let $\strata$ denote the set of all strata under
consideration, initialized to $\{1,2,\ldots,r\}$. Let $\mathcal{O}$ denote the set of oversized samples, under target memory budget for
$\mathbb{S}$, and $\mathcal{U} = \mathbb{S} - \mathcal{O}$ denote the collection of samples that are not oversized.
When the context is clear, we use $\mathcal{O}, ~ \mathcal{U}$, and $\strata$ to refer to the set of stratum identifiers as well as the
set of samples corresponding to these identifiers.
 
\begin{lemma}
\label{lem:recur}
A variance-optimal eviction of $\beta$ elements from $\mathbb{S}$
under memory budget $M'$ requires a variance-optimal eviction of $\beta$ elements from
$\mathcal{O}$ under memory budget $M'-\sum_{j \in \mathcal{U}} s_j$.
\end{lemma}
\begin{proof}
%
  Recall that $s'_i$ denotes the final size of sample $S_i$ after
  $\beta$ elements are evicted. Referring to the variance $V$ from
  Equation~\ref{eqn:var}, we know a variance-optimal sample size
  reduction  of $\beta$ elements from
  $\mathbb{S}$ under memory budget $M'$ requires to minimize
\begin{eqnarray}
\sum_{i\in\mathcal{A}} \frac{n_i^2 \sigma_i^2}{s'_i}
- \sum_{\in\mathcal{A}}\frac{n_i^2 \sigma_i^2}{s_i}\label{eqn:voe1}
\end{eqnarray}

By Lemma~\ref{lem:evict-o}, we know $s_i = s'_i$ for all
$i \in \mathcal{U}$. Hence, minimizing Formula~\ref{eqn:voe1} is
equivalent to minimizing
\begin{eqnarray}
  \sum_{i=\mathcal{O}} \frac{n_i^2 \sigma_i^2}{s'_i} - \sum_{i\in \mathcal{O}} \frac{n_i^2 \sigma_i^2}{s_i}\label{eqn:voe2}
\end{eqnarray}

The minimization of Formula~\ref{eqn:voe2} is exactly the result
obtained from a variance-optimal sample size reduction 
of $\beta$ elements from oversized samples under
the new memory budget $M'-\sum_{i \in \mathcal{U}} s_i$.
\end{proof}


\begin{observation}
\label{ob:exit}
In the case every sample in the \srs\ is oversized under target memory $M'$, \ie $\mathbb{S} = \mathcal{O}$, the variance-optimal reduction is to reduce the size of each sample $S_i \in \mathbb{S}$ to its Neyman allocation $M'_i$ under the new memory budget $M'$.
\end{observation}

The following theorem summarizes the correctness and time complexity of Algorithm \ssr. 

\begin{theorem}
\label{thm:ssr}
Algorithm~\ref{algo:reduction-batch} ($\ssr$) finds a variance-optimal reduction of the stratified random sample $\strata$ under new memory budget $M$. The worst-case time of $\ssr$ is $O(r^2)$, where $r$ is the number of strata.
\end{theorem}

\begin{proof}
Correctness follows from Lemmas~\ref{lem:evict-o}--\ref{lem:recur} and Observation~\ref{ob:exit}. The worst-case time happens when each recursive call sees only one stratum that is not oversized. In such a case, the time of all recursions of $\ssr$ on a \srs\ across $r$ strata is: $O(r+(r-1)+\ldots+1) = O(r^2)$.
\end{proof}

Although $\ssr$ takes $O(r^2)$ time in the worst case, its time complexity tends to be much better in practice.  
If the number of samples that are not oversized contributes at least~a certain percentage of the total number of samples being considered in every recursion, its overall time cost will be $O(r)$.

\subsubsection{A faster implementation}
We also present an iterative algorithm for sample size reduction, \fssr,  with time complexity $O(r\log r)$. \fssr\ shares the same algorithmic foundation as $\ssr$, but uses a faster method to find samples that are not oversized. 


\remove{
\begin{theorem}
\label{thm:fssr}
There is an algorithm $\fssr$ for variance-optimal sample size reduction on $r$ strata, whose worst-case time complexity is $O(r\log r)$.
\end{theorem}
}

\newcommand{\sumsofar}{D}

\begin{algorithm}[t]
\DontPrintSemicolon
\caption{\fssr($M$): A fast implementation of Sample Size Reduction without using recursion.}
\label{algo:reduction-batch-fast}

\KwIn{The strata under consideration is implicitly
  $\{1,2,\ldots,r\}$. $M$ is the target total sample size.}
\KwOut{For $1 \le i \le r$, $\mathcal{L}[i]$ is set to the final size
  of sample for stratum $i$, such that the increase of the variance
  $V$ is minimized.}

Allocate $\mathcal{L}[1..r]$, an array of numbers\; 

Allocate  $Q[1..r]$, an array of $(x,y,z)$ tuples\;
\lFor{$i=1 \ldots r$}{
  $Q[i] \gets (i, n_i \sigma_i, s_i / (n_i \sigma_i))$;
}

Sort array $Q$ in ascending order on the $z$ dimension\;\label{line:fssr-sort} 

\For{$i= (r-1)$ down to $1$}{
  $Q[i].y \gets Q[i].y + Q[i+1].y$
}

\medskip
$M_{new}\gets M$;
$\sumsofar \gets Q[1].y$\;

\For{$i=1\ldots r$}{
    $M_{Q[i].x} \gets M \cdot n_{Q[i].x}\sigma_{Q[i].x} / \sumsofar$\;
    \lIf{$s_{Q[i].x} > M_{Q[i].x}$}{
      break
    } 
    $\mathcal{L}[Q[i].x] \gets s_{Q[i].x]}$\;
    $M_{new} \gets M_{new}-s_{Q[i].x}$\;

\medskip 

    \tcp{Check the next sample, which must exist.}
    $M_{Q[i+1].x} \gets M \cdot n_{Q[i+1].x}\sigma_{Q[i+1].x} / \sumsofar$\;
    \If(\tcp*[h]{oversized}){$s_{Q[i+1].x} > M_{Q[i+1].x}$}{
      $M \gets M_{new}$;
      $\sumsofar \gets Q[i+1].y$\;
    }
 }

\medskip 

\tcp{Reduce sample size to target.}
\For{$j=i..r$}{
  \tcp{Desired size for $S_{Q[j].x}$}
  $\mathcal{L}[{Q[j].x}] \gets M  \cdot
  n_{Q[j].x}\sigma_{Q[j].x}/\sumsofar$\;
}

\Return{$\mathcal{L}$}\;
\end{algorithm}

\begin{definition}
Let $Q[1..r]$ be an array of $(x,y,z)$ tuples, where each $Q[i]$ is
initialized as $(i, n_i\sigma_i, s_i/(n_i\sigma_i))$.  
Array $Q$ is then sorted on its $z$ dimension.
\end{definition}

\begin{lemma}
\label{lem:fssr-obs}
Under any given memory budget $M$, if there exists at least one
unoversized sample, the collection of the identifiers of the
unoversized samples must be occupying a continuous prefix of the array
$Q$.
\end{lemma}

\begin{proof}
  Recall that under a memory budget $M$, the Neyman allocation size
  for stratum $i$ is $M_i = n_i\sigma_i / \sumsofar$, where
  $\sumsofar = \sum_{i=1}^r n_j\sigma_j$.  A sample $S_i$ is not
  oversized if and only if $s_i \leq M_i$, i.e.,
  $s_i/(n_i\sigma_i) \leq 1/\sumsofar$.  A sample $S_i$ is oversized
  if and only if $s_i > M_i$, i.e., $s_i/(n_i\sigma_i) > 1/\sumsofar$.
  Because array $Q$ is in the ascending order of its $z$ dimension,
  the lemma is proved.
\end{proof}

Lemma~\ref{lem:fssr-obs} implies that we can linearly walk along the
array $Q$ from $Q[1]$ toward $Q[r]$.  By comparing the sample size and
the Neyman allocation size for each stratum we are looking at during
the walk, we will be able to find the collection of samples that are
not oversized, under the new target memory budget $M'$.  

After finding the
prefix of the $Q$ array that represents the collection of samples that
are not oversized, we pause the walk and then set the new memory $M'$
budget to be $M'$ minus the total size of the samples in the
prefix. 
Then, we treat the remaining part (after
excluding the prefix) of the array $Q$ as the current array $Q$ and do the
same walk under the new memory budget $M'$.

The walk will stop if we do not see any sample that is not oversized
under the current memory budget $M'$. In that case, we just set the size of the
samples in the current array $Q$ to be their Neyman
allocation size, under the current memory budget.

In order to avoid the recomputation of $\sumsofar$, which is needed in
computing the Neyman allocation, for every new memory budget during
the walk, we precompute the $\sumsofar$ for every suffix of the array
$Q$ and save the result in the $y$ dimension of the $Q$ array.

The method $\fssr$ in Algorithm~\ref{algo:reduction-batch-fast}
shows the pseudocode of this faster algorithm for variance-optimal
sample size reduction.

\begin{theorem}
\label{thm:fssr}
(1) The $\fssr$ procedure in Algorithm~\ref{algo:reduction-batch-fast}
finds the correct size of each sample of an \srs, whose memory
budget is reduced to $M$, 
such that the increase of the variance $V$ is minimized. 
(2) The worst-case time cost of $\fssr$ on a \srs\  across $r$ strata
is $O(r\log r)$.
\end{theorem}

\begin{proof}
  (1) The correctness of the procedure follows from
  Lemmas~\ref{lem:evict-o}--\ref{lem:recur},
  Observation~\ref{ob:exit}, and Lemma~\ref{lem:fssr-obs}.
(2) The time cost of $\fssr$ is dominated by the step of sorting array
$Q$ on its $z$ dimension (Line~\ref{line:fssr-sort}), so the
worst-case time cost of $\fssr$ is $O(r\log r)$.
\end{proof}

%% file: offline.tex
\section{{\Large \voila}: Variance-Optimal Offline~SRS}
\label{sec:offline}
We now present an algorithm for computing the variance-optimal allocation of sample sizes in the case when one or more strata may be bounded. Note that the actual sampling step is straightforward for the offline algorithm -- once the allocation of sample sizes is determined, the samples can be chosen in a second pass through the data, using reservoir sampling within each stratum. Hence, in the rest of this section, we focus on determining the allocation. Consider a static data set $R$ of $n$ elements across $r$ strata, where stratum $i$ has $n_i$ elements, and has standard deviation $\sigma_i$. \emph{How can a memory budget of $M$ elements be partitioned among the strata in a variance-optimal manner?}
We present \voila ({\bf V}ariance-{\bf O}pt{\bf I}ma{\bf L} {\bf A}llocation), an efficient offline algorithm for variance-optimal allocation
that can handle strata that are bounded.  \voila~ is a generalization of the classic Neyman allocation -- 
in the case when every stratum has abundant data, it reduces to Neyman allocation.

\remove{
The idea is as follows. Consider the
expression for the variance $V$ in Equation~\ref{eqn:var}. Since parameters $n$, $r$, $n_i$, and
$\sigma_i$, are constants that cannot be affected by this decision,
minimizing $V$ only requires the minimization of
$\sum_{i=1}^r ({n_i^2\sigma_i^2}/{s_i})$. This leads to the following
optimization problem.
\medskip 
\noindent
\begin{eqnarray}
\text{Minimize } \sum_{i=1}^r \frac{n_i^2\sigma_i^2}{s_i}\label{eqn:obj2}
\end{eqnarray}
\noindent
subject to constraints: 
\begin{eqnarray}
0\leq s_i \leq n_i,  \textrm{ for each }i=1,2,\ldots, r\label{eqn:cons2.1}\\
\sum_{i=1}^r s_i = M\label{eqn:cons2.2}
\end{eqnarray}

The similarity in structure between this optimization problem and \vor, formulated in~(\ref{eqn:obj1})--(\ref{eqn:cons1.2}).
}

The following two-step process reduces variance-optimal offline SRS to variance-optimal sample size reduction.

\smallskip
\noindent {\bf Step 1:} Suppose we start with a memory budget of $n$. Then, we will just save the whole data set in the \srs, and thus each sample size $s_i = n_i$. By doing so, the variance $V$ is minimized, since $V=0$ (Equation~\ref{eqn:var}).

\smallskip
\noindent {\bf Step 2:} Given the \srs\ from Step~1, we reduce the memory budget from $n$ to $M$ such that the resulting variance is the smallest. This can be done using variance-optimal sample size reduction, by calling  $\ssr$ or $\fssr$ with target sample size $M$.



\smallskip
\noindent \voila~(Algorithm~\ref{algo:offline}) 
simulates this process. The algorithm only records the sample
sizes of the strata in array $\mathcal{L}$, without creating the actual samples. 
The actual sample from stratum $i$ is created by choosing
$\mathcal{L}[i]$ random elements from stratum $i$, using any method
for offline uniform random sampling without replacement.

\begin{algorithm}[t]
\DontPrintSemicolon
\caption{\voila($M$): Variance-optimal stratified random sampling for bounded data}
\label{algo:offline}
\KwIn{$M$ is the memory target}

  \For{$i=1\ldots r$\label{line:off-for0}}{
    $s_i \gets n_i$\label{line:off-for1} ~\tcp{assume total available memory of $n$}
}

  $\mathcal{L} \gets \fssr (M)$\;\label{line:off-ssr1} 

\Return{$\mathcal{L}$}    \tcc*{$\mathcal{L}[i] \le n_i$ is the sample size for stratum $i$ in a  variance-optimal \srs \label{line:off-ssr}.} 
%
\end{algorithm}

\begin{theorem}
\label{thm:bound}
Given a data set $\stream$ with $r$ strata, and a memory budget $M$, $\voila$ (Algorithm~\ref{algo:offline}) returns in $\mathcal{L}$ the sample size of each stratum in a variance-optimal stratified random sample. The worst-case time cost of $\voila$ is $O(r\log r)$.
\end{theorem}

\begin{proof}
The correctness follows from the correctness of Theorem~\ref{thm:fssr}, since the final sample is the sample of the smallest variance that one could obtain by reducing the initial sample (with zero variance) down to a target memory of size $M$. The run time is dominated by the call to $\fssr$, whose time complexity is $O(r\log r)$.
\end{proof}

%% file: online.tex
\section{Streaming SRS}
\label{sec:streaming}
We now consider the maintenance of an SRS from a data stream, whose elements are arriving continuously.


\subsection{A Lower Bound for Streaming SRS}
\label{sec:lb}
Given a data stream $\stream$ across $r$ strata, let $V^*$ denote the
sample variance of the stratified random sample created by \voila, 
 using a memory budget of $M$.  Because \voila\ is variance optimal,
$V^*$ is the smallest variance that we can get from any stratified
random sample of $\stream$ under the memory budget $M$. While \voila~
is not a streaming algorithm, $V^*$ is a lower bound on the variance that
a streaming algorithm can achieve, under memory budget $M$.

Let $V$ denote the sample variance of an SRS of $\stream$ using the
same memory budget $M$. We say $V$ is an approximation of $V^*$ with a
multiplicative error of $\alpha$, for some constant $\alpha \geq 1$,
if: (1)~the sample within each stratum $i$ is chosen uniformly at
random without replacement from stratum $i$.  (2)~$V \leq \alpha \cdot V^*$.

\begin{theorem}
\label{thm:lb-var}
Any streaming algorithm for maintaining an SRS over a stream with $r$ strata using a memory of $M$ records must, in the worst case, have a multiplicative error $\Omega(r)$ when compared with the optimal variance that can be achieved by a stratified random sample using memory of $M$ records.
\end{theorem}


\begin{proof}
We use proof by contradiction.  Suppose that it is possible to
maintain an approximate stratified random sample with a multiplicative
error less than $r$.

Consider an input stream where the $i$th stratum consists of elements
in the range $[i,i+1)$, where the right endpoint of the stratum does
not include $i+1$. Suppose the stream so far has the following
elements. For each $i$ from $1$ to $r$, there are $(\alpha-1)$ copies
of element $i$ and one copy of $(i+\eps)$ where $0<\eps <1$ and
$\alpha\geq 3$. After observing these
elements, for each stratum $i$, $1\leq i \leq r$, we have:

$$
n_i = \alpha, \ \ \ \ \ \ \mu_i = i + \frac\eps\alpha,
$$  
$$
\sigma_i =
\sqrt{\left((\alpha-1)\left(\frac{\eps}{\alpha}\right)^2+\left(\eps-\frac{\eps}{\alpha}\right)^2\right)\Big/\alpha}=\frac{\sqrt{\alpha-1}}{\alpha}\eps. 
$$  

Observe that, due to the memory budget $M$, at least one stratum has
its sample size no more than $M/r$. Without loss of generality, 
let's say that stratum is stratum $1$.

Suppose an element of value $(2-\eps)$ arrives in the stream, where
$\eps = 1/(r-1)$. This element belongs
to stratum $1$.  
Let $n_1'$, $\mu_1'$, and $\sigma_1'$
denote the new size, mean, and standard deviation of stratum $1$ after
this element arrives.


$$
n'_1= \alpha+1, \ \ \ \ \ \  \mu'_1 = 1 + \frac{1}{\alpha+1},
$$
\begin{eqnarray*}
\sigma'_1 = \sqrt{\frac{(\alpha-1)\left(\frac{1}{\alpha+1}\right)^2+\left(\eps-\frac{1}{\alpha+1}\right)^2+\left(1-\eps-\frac{1}{\alpha+1}\right)^2}{\alpha+1}}
=\sqrt{\frac{\eps^2 + (1-\eps)^2 - \frac{1}{\alpha+1}}{\alpha+1}}.
\end{eqnarray*}

\noindent
It follows that:
\begin{eqnarray}
  \label{eq:6}
(\alpha + 1)\sqrt{\frac{\frac 1 2 - \frac{1}{\alpha+1}}{\alpha+1}}
\leq n'_1\sigma'_1 
\leq (\alpha + 1)\sqrt{\frac{1- \frac{1}{\alpha+1}}{\alpha+1}}\label{eqn:6-1}\\
\Longrightarrow
\frac{\sqrt{\alpha}}{2}
\leq n'_1\sigma'_1 
\leq \sqrt{\alpha} \ \ \ \ \ \ \ \ (\textrm{Note: $\alpha>2$})\label{eqn:6-3}
\end{eqnarray}
In~\ref{eqn:6-1}, the left inequality stands when $\eps = 1/2$ and the right inequality
stands when $\eps=0$ or $1$. 
We also have: 
\begin{eqnarray*}
  \label{eq:7}
  \sum_{i=2}^r n_i \sigma_i = (r-1) \alpha
  \frac{\sqrt{\alpha-1}}{\alpha} \eps 
= \sqrt{\alpha-1}  \ \ \ \ \left(\textrm{Note: $\eps = \frac{1}{r-1}$}\right)\notag
\end{eqnarray*}
\begin{eqnarray}
\Longrightarrow
\frac{\sqrt{\alpha}}{2} 
\leq 
\sum_{i=2}^r n_i \sigma_i 
\leq \sqrt{\alpha}  \ \ \ \ \ \ \ \ (\textrm{Note: $\alpha>2$})\label{eqn:7-2}
\end{eqnarray}

Let $V$ denote the sample variance of the
stratified random sample maintained over the stream of $(r\alpha +
1)$ elements.  Let $V^*$ denote the smallest sample
variance that one can get from a stratified random sample from these
$(r\alpha+ 1)$ date elements.
Let $\Delta = \left(n'_1{\sigma'}_1^2 + \sum_{i=2}^r n_i\sigma_i^2
\right)\big/n^2$.

We observe the facts that 
(1) after processing these $(r\alpha + 1)$ elements,  the sample size
$s_1\leq M/r + 1$.  (2) The portion of the sample variance contributed
by strata $2,3,\ldots,r$ is minimized if the memory budget for these
strata, which is no more than $M$, are equally shared, because all
$n_i\sigma_i$ are equal for $i=2,3,\ldots, r$. Using these two facts
and the definition of the sample variance in equation~\ref{eqn:var},
we have:
\begin{eqnarray*}
  \label{eq:1}
V &=& \frac{1}{n^2} \left(\frac{{n'}_1^2{\sigma'}_1^2}{s_1}+
\sum_{i=2}^r \frac{n_i^2 \sigma_i^2}{s_i} \right)-  \Delta\\
&\geq&  \frac{1}{n^2} \left(\frac{{n'}_1^2 {\sigma'}_1^2}{M/r+1} 
+ \sum_{i=2}^r \frac{n_i^2 \sigma_i^2}{M/(r-1)} \right)- \Delta\\
&\geq&  \frac{1}{n^2} \left(\frac{\alpha/4}{M/r+1} 
+ \sum_{i=2}^r \frac{(\alpha-1)\eps^2}{M/(r-1)} \right)- \Delta\\
&=&  \frac{1}{n^2} \left(\frac{\alpha/4}{M/r+1} 
+ \frac{\alpha-1}{M} \right)- \Delta  
\ \ \ \ \left(\textrm{Note: $\eps = \frac{1}{r-1}$}\right)
\end{eqnarray*}

On the other hand, the smallest sample variance $V^*$ 
is achieved by
using the Neyman allocation of the memory budget $M$, assuming each
stratum has sufficient data to fill its sample size assigned by the
Neyman allocation. By Inequalities~\ref{eqn:6-3} and~\ref{eqn:7-2}, we
know that in the Neyman allocation for the current stream of $r\alpha
+1$ elements, stratum $1$ uses
at least $M/3$ memory space, whereas all other strata equally share
at least $M/3$ memory space as well because all $n_i\sigma_i$ are
equal for $i=2,3,\ldots, r$. Using these observations into 
Equation~\ref{eqn:var}, we have:
\begin{eqnarray*}
  \label{eq:2}
V^* &\leq &  
\frac{1}{n^2} \left(\frac{{n'}_1^2 {\sigma'}_1^2}{M/3} 
+ \sum_{i=2}^r \frac{n_i^2 \sigma_i^2}{M/3(r-1)} \right)
- \Delta\\
&\leq&
\frac{1}{n^2} \left(\frac{\alpha}{M/3} 
+ \sum_{i=2}^r \frac{(\alpha-1)\eps^2}{M/3(r-1)} \right)
- \Delta\\
&=&
\frac{1}{n^2} \frac{6\alpha-3}{M}
- \Delta \ \ \ \ \ \ \ \ \left(\textrm{Note: $\eps = \frac{1}{r-1}$}\right)
\end{eqnarray*}

Because $\Delta \ge 0$ and $M > r$, we have:

\begin{eqnarray*}
  \label{eq:3}
  \frac{V}{V^*} \geq\frac{V+\Delta}{V^*+\Delta} 
\geq 
\Omega(r)
\end{eqnarray*} 
\end{proof}

The idea in the proof is to construct an input stream with $r$ strata where the variance of all strata 
are the same until a certain point, where the variance of a single stratum increases to a high value -- a variance-optimal
SRS will respond by increasing the allocation to this stratum. However, a streaming algorithm is unable to do so quickly,
since it is in general unable to collect enough samples to satisfy the increased allocation to this stratum. Though a streaming
algorithm is able to compute the variance-optimal allocation to different strata in an online manner, it cannot 
actually maintain these samples using limited memory. 

We also note that the above lower bound is tight, since the simple uniform allocation, which allocates 
$M/r$ memory to each of the $r$ strata that have been observed so far, has a variance which is within a 
multiplicative factor of $r$ of the optimal. However, we see that the policy of uniform allocation performs 
poorly in practice, since it does not distinguish between different strata, whether based on volume or variance.

\subsection{{\Large \svoila}: Practical Streaming SRS}
\label{sec:svoila}
We now present $\svoila$, a practical streaming algorithm for stratified random sampling, that works for streams with zero or more bounded strata. Choices made by \svoila~ are "locally optimal" in the following sense: when new stream elements arrive, the decision of whether or not to select these elements (which will make it necessary to discard sampled elements from other strata) is made in a way that minimizes the variance of the estimate from resulting sample. \svoila~ can be viewed as an online version of \voila, which constructs an SRS with minimal variance using a multi-pass algorithm through the entire data. 


Let $\stream$ denote the stream so far, and $\stream_i$ the substream of elements belonging to stratum $i$. Within a single stratum, any algorithm for SRS needs to maintain a uniform random sample of all data seen so far. In streaming SRS, the memory $s_i$ allocated to a stratum $i$ may change with time, depending on the data arriving within this stratum, and other strata. One issue for a streaming algorithm is to maintain a uniform random sample within stratum $i$ when $s_i$ is changing. A decrease in the allocation $s_i$ can be handled easily, through discarding randomly chosen elements from the current sample $S_i$ until the desired sample size is reached. What if we need to increase the allocation to stratum $i$? If we simply start sampling new elements according to the higher allocation to stratum $i$, then recent elements in the stream will be favored over the older ones, and the sample within stratum $i$ is no longer uniformly chosen. In order to ensure that $S_i$ is always chosen uniformly at random from $\stream_i$, newly arriving elements in $\stream_i$ need to be held to the same sampling threshold as older elements, even if the allotted sample size $s_i$ increases. 

$\svoila$ maintains sample $S_i$ as follows. An arriving element from $\stream_i$ is assigned a random "key" drawn uniformly from the interval $(0,1)$. The algorithm maintains the following invariant: {\em $S_i$ is the set of $s_i$ elements with the smallest keys among all elements so far in $\stream_i$.} Note that this means that if we desire to increase the allocation to stratum $i$, then this may not be accomplished immediately, since a newly arriving element in $\stream_i$ may not be assigned a key that meets this sampling threshold. Instead, the algorithm has to wait until it receives an element in $\stream_i$ whose assigned key is small enough.  In order to ensure the above invariant, the algorithm maintains, for each stratum $i$, a variable $d_i$ that tracks the smallest key of an element in $\stream_i$ that is not currently included in $S_i$. If an arriving element in $\stream_i$ has a key that is smaller than or equal to $d_i$, it is
included within $S_i$; otherwise, it is not. 

\begin{algorithm}[t]
\caption{$\svoila$: Initialization}
\label{algo:ssna-init}
\KwIn{$M$ -- total sample size, $r$ -- number of strata.}
\tcp{$S_i$ is the sample for stratum $i$, and $\stream_i$ is the substream of elements from Stratum $i$}

Load the first $M$ stream elements in memory, and partition them into $r$ per-stratum samples, $S_1, S_2, \ldots, S_r$, such that $S_i$ consists of $(e,d)$ tuples from stratum $i$, where $e$ is the element, $d$ is the key of the element, chosen uniformly at random from $(0,1)$.

For each stratum $i$, compute $n_i$, $\sigma_i$. Initialize $d_i \gets 1$, where $d_i$ is the smallest key among all elements in $\stream_i$ not selected in $S_i$.
\end{algorithm}

\begin{algorithm}[t]
  \caption{$\svoila$: Process a new minibatch $B$ of $b$ elements. Note that $b$ need not be fixed, and can vary from one minibatch to the other.}
\label{algo:ssna-batch}

$\beta \gets 0$\tcp*{\#selected elements in the minibatch}

\For{each $e\in B$\label{line:for}}{
    Let $\alpha$ denote the stratum of $e$
    
    Update $n_\alpha$ and $\sigma_\alpha$
    
    Assign a random key $d\in (0,1)$ to element $e$\;

    \If(\tcp*[f]{element $e$ is selected}){$d \leq d_\alpha$}{
      $S_\alpha \leftarrow \{e\} \bigcup S_\alpha$;
      $\beta \gets \beta + 1$\label{line:beta}\;
    }
}

\tcc{Variance-optimal eviction of $\beta$ elements}

\If(\tcp*[f]{faster for evicting 1 element}){$\beta =1 $\label{line:svoila-evict-if-start}}{
$\ell \gets \sssr()$\;
      \hspace*{-0.5mm}Delete one element of largest key from $S_\ell$\;
      $d_{\ell} \gets $ smallest key discarded from $S_{\ell}$\;\label{line:svoila-evict-if-end}
}
\ElseIf{$\beta > 1$\label{line:svoila-evict-else-start}}{
$\mathcal{L} \gets \fssr(M)$\; 

  \For(\tcp*[f]{Actual element evictions}){$i=1\ldots r$}{
    \If{$\mathcal{L}[i] < s_i$}{
      Delete $s_i- \mathcal{L}[i]$ elements of largest keys from~$S_i$\;
      $d_{i} \gets $ smallest key discarded from $S_{i}$\;\label{line:svoila-evict-else-end}
    }
  }
}
\end{algorithm}

Algorithm~\ref{algo:ssna-init} presents the initialization of $\svoila$, which simply loads the first $M$ stream elements into the memory budget and divides them into $r$ samples $S_1, S_2, \ldots,S_r$, and initializes state. As new elements arrive, they change the frequency and the variance of a stratum and may lead to changes in the desired allocation of samples to strata. While it is possible to recompute the variance-optimal allocation, it is not possible to sample additional elements into strata as necessary, since we do not have the ability to look at all the data seen so far. However, our algorithm locally optimizes the variance through carefully selecting the strata from which samples will be discarded to make way for one or more incoming sampled elements.

\svoila\ supports the insertion of a minibatch of any size $b\geq 1$, where the value of $b$ is even allowed to be dynamic during the execution of \svoila.  When users fix $b=1$, \svoila\ becomes streaming algorithm that handles one element at a time. As the value $b$ increases, we can expect \svoila\ to have a better variance, since its optimization decisions are based on greater amount of data. Algorithm~\ref{algo:ssna-batch} presents the algorithm for maintaining the stratified random sample when a new minibatch of multiple elements arrives. Lines~\ref{line:for}--\ref{line:beta} make one pass through the minibatch to update the statistics of each stratum and store the selected elements into the sample.  If $\beta>0$ elements from the minibatch get selected into the sample, in order to balance the memory budget at $M$,  we will need to evict $\beta$ elements from the \srs -- this is accomplished using the variance-optimal sample size reduction technique from Section~\ref{sec:reduction}.  For the special case where we only need to evict one element, we can use the faster algorithm $\sssr$ (Lines~\ref{line:svoila-evict-if-start}--\ref{line:svoila-evict-if-end}); otherwise, $\fssr$ is used (Lines~\ref{line:svoila-evict-else-start}--\ref{line:svoila-evict-else-end}).

Lemma~\ref{lem:svoila-uniform} shows that the sample maintained within by \svoila within each stratum is a uniform random sample, showing this is a valid stratified sample. 


\begin{lemma}
\label{lem:svoila-uniform}
For each $i=1,2,\ldots, r$ sample $S_i$ maintained by \svoila~(Algorithm~\ref{algo:ssna-batch}) is selected uniformly at random without replacement from stratum $R_i$.
\end{lemma}

\begin{proof}
  First, note that each $S_i$ is selected from $\stream_i$ without replacement,
  because each element of $\stream_i$ is selected into $S_i$ no more than once.
  Next, we prove the uniformity of $S_i$. In case $|S_i|=n_i$, all
  elements of $\stream_i$ are in $S_i$. In case $|S_i| < n_i$, $S_i$
  contains the $|S_i|$ elements with the smallest keys from stratum $\stream_i$,
  because: (1)~Anytime an element is discarded from $S_i$, it is the
  element of the largest key in the sample. (2)~If another element
  with key $d$ enters later, it cannot be inserted into $S_i$ unless
  $d$ is smaller than or equal to  all other keys discarded so far.
  Because the keys of elements are assigned randomly, each element has
  a chance of $|S_i|/n_i$ to be selected into $S_i$. Therefore, $S_i$
  is a uniform random sample from $\stream_i$ without replacement.
\end{proof}

\begin{theorem}
\label{thm:batch-time}
If the minibatch size $b=1$, then the worst-case time cost of \svoila\ for processing an element is $O(r)$. 
The expected time for processing an element belonging to stratum $\alpha$ is $O(1 + r \cdot s_\alpha/n_\alpha)$, 
which is $O(1)$ when $r\cdot s_\alpha = O(n_\alpha)$.
If $b> 1$, then the worst-case time cost of \svoila\ for processing a minibatch is $O(r\log r+b)$. 
\end{theorem}

\begin{proof}
  $b=1$: The worst case happens when the single new element
  from belonging to stratum $\alpha$ gets selected into $S_\alpha$. In
  that case, we need to reduce the \srs\ size by one via \sssr, which
  takes $O(r)$ time. The probability that the new element is selected
  into $S_\alpha$ is equal to $s_\alpha/n_\alpha$, so the expected
  time follows.

 $b> 1$:  The time cost for Lines~\ref{line:for}--\ref{line:beta} is
  $O(b)$. The time cost for
  Lines~\ref{line:svoila-evict-if-start}--\ref{line:svoila-evict-else-end}
  is $O(r\log r +\beta)$. So the total time cost is
  $O(b) + O(r\log r + \beta) = O(r\log r + b)$.
  The per-element amortized time cost is $O(1)$ when $b = \Omega(r\log r)$
\end{proof}

\noindent
We can expect \svoila\ to have an amortized per-item processing time of $O(1)$ in many circumstances.

When $b=1$: After observing enough stream elements
from stratum $\alpha$, such that $r\cdot s_\alpha = O(n_\alpha)$, 
the expected processing time of an element becomes
$O(1)$. Even if certain strata have a very low frequency, the
expected time cost for processing  a single element 
is still expected to be $O(1)$, because elements from an infrequent
stratum  $\alpha$ are unlikely to appear in the minibatch.  

When $b> 1$: The per-element amortized time cost of \svoila\ 
is $O(1)$, when the minibatch size $b=\Omega(r\log r)$. 

%% file: experiment.tex
\section{Experimental Evaluation}
\label{sec:eval}
The algorithms are evaluated on real-world data as well as synthetic data. The input is a stored set or a continuous stream of records from a data source, which is processed by the sampler which either outputs the sample at the end of computation (offline sampler) or continuously maintains a sample (streaming sampler). A streaming sampler must process data in a single pass, and is unable to access elements that were observed earlier, unless they are stored in memory. An offline sampler has access to all data received, and can compute a stratified random sample using multiple passes through data. 

We evaluate the algorithms in two ways. The first is a direct evaluation of the quality of the samples, through the resulting allocation and the variance of an estimate of the population mean obtained using the samples. The second is through the accuracy of approximate query processing using the maintained samples.


\subsection{Sampling Methods}
We implemented three offline sampling methods, each of which uses two passes to compute a stratified random sample. 
The first pass is to determine strata characteristics from which the sample size of each stratum is derived, and the second pass is to collect the samples. Each method is given the same total memory of $M$ records. 
We implemented \voila as described in the paper, and \neyman, that uses Neyman allocation.
As explained earlier, \neyman will lead to a bounded stratum being allocated a greater sample size than the data within the stratum. 
This leaves some portion of the total memory unused by \neyman.
To improve upon this, we implemented an extended version of \neyman called \neymanPlus,
which uses the entire memory allocation. \neymanPlus first runs \neyman. Any unused memory is allocated equally 
among the remaining (non-bounded) strata. This may lead to more strata becoming bounded, 
and the process is continued recursively, until all the memory is used up.


We implemented the following stream sampling methods: \svoila with different minibatch sizes, reservoir sampling \reservoir, and \ssunif\ -- SRS with uniform allocation. 
\reservoir maintains a uniform random sample chosen without replacement from the stream 
- we expect the number of samples allocated to stratum $i$ by \reservoir to be proportional to $n_i$. 
\ssunif allocates the same amount of memory to each stratum that has been observed. If a stratum has~too few
data points to fill its current allocation, then the remaining memory is allocated uniformly among other strata, and this memory
redistribution may happen further, recursively.

For all experiments on comparing sampling methods based on allocations or on variance, each data point is the mean of five independent runs.

\subsection{Data}
\begin{figure}[t]
	\centering
	\begin{subfigure}{\graphwidth}
		\includegraphics[width=\textwidth]{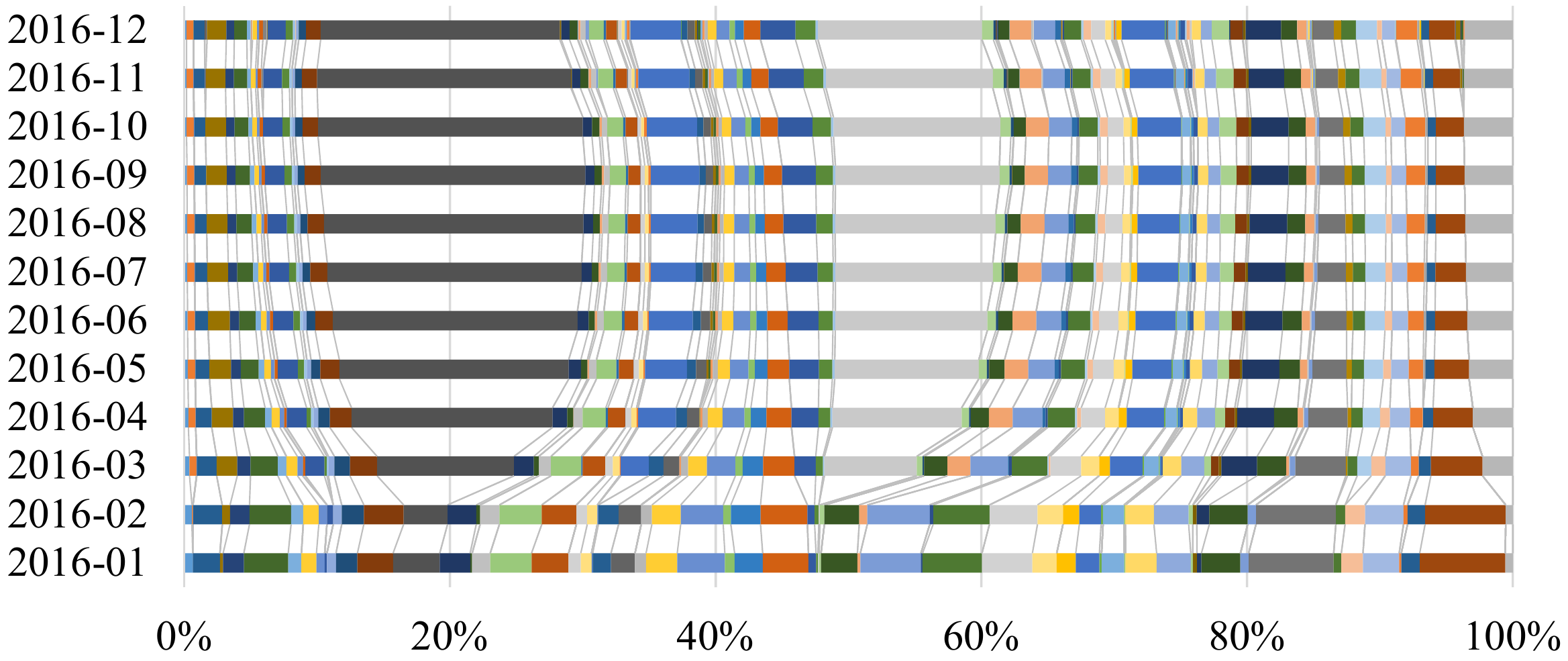}
		\caption{Relative frequencies of different strata. The x-axis is the fraction of points observed so far. At different points in time, the relative (cumulative) frequency of each stratum is shown.}
		\label{fig:data_freq}
	\end{subfigure}
	~
	\begin{subfigure}{\graphwidth}
		\includegraphics[width=\textwidth]{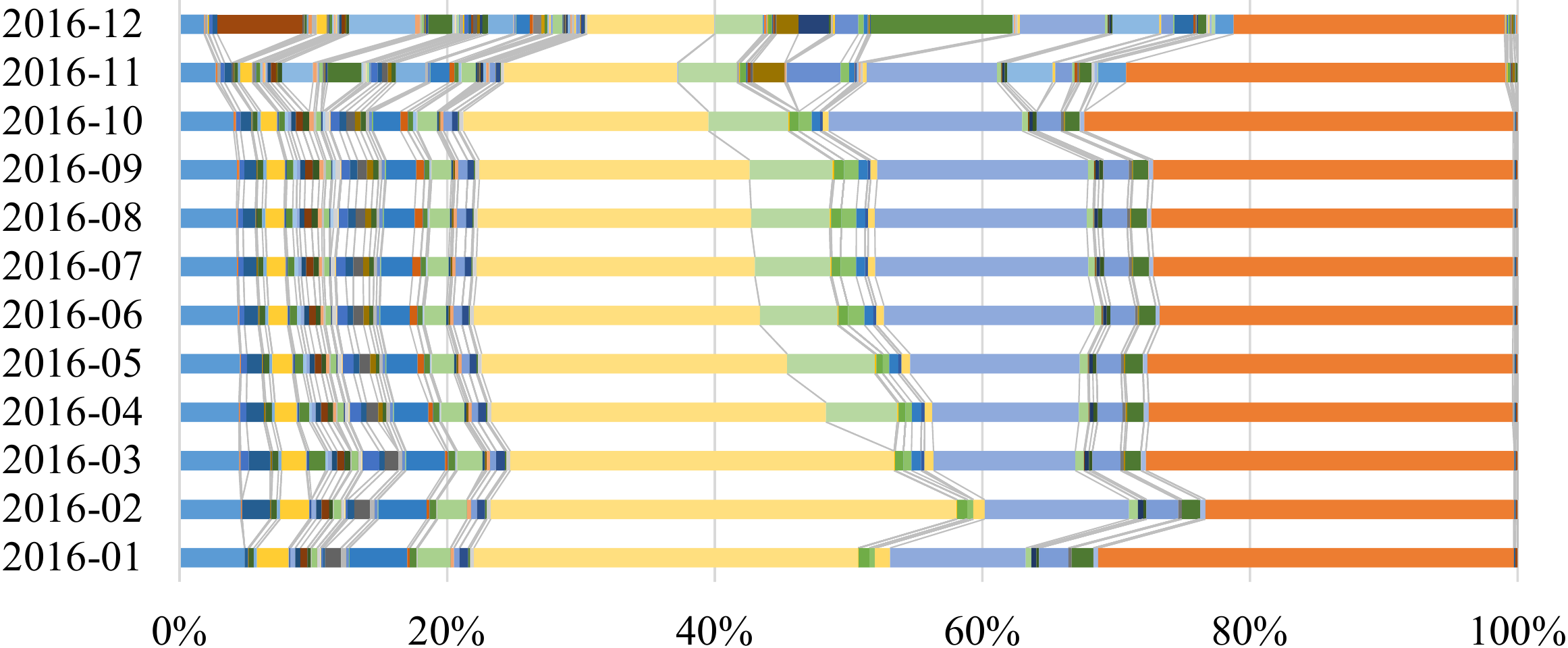}
		\caption{Relative standard deviations of different strata, demonstrated by normalized cumulative standard deviations observed by the end of each month.}
		\label{fig:data_stddev}
	\end{subfigure}
	~
	\begin{subfigure}{\graphwidth}
		\includegraphics[width=\textwidth]{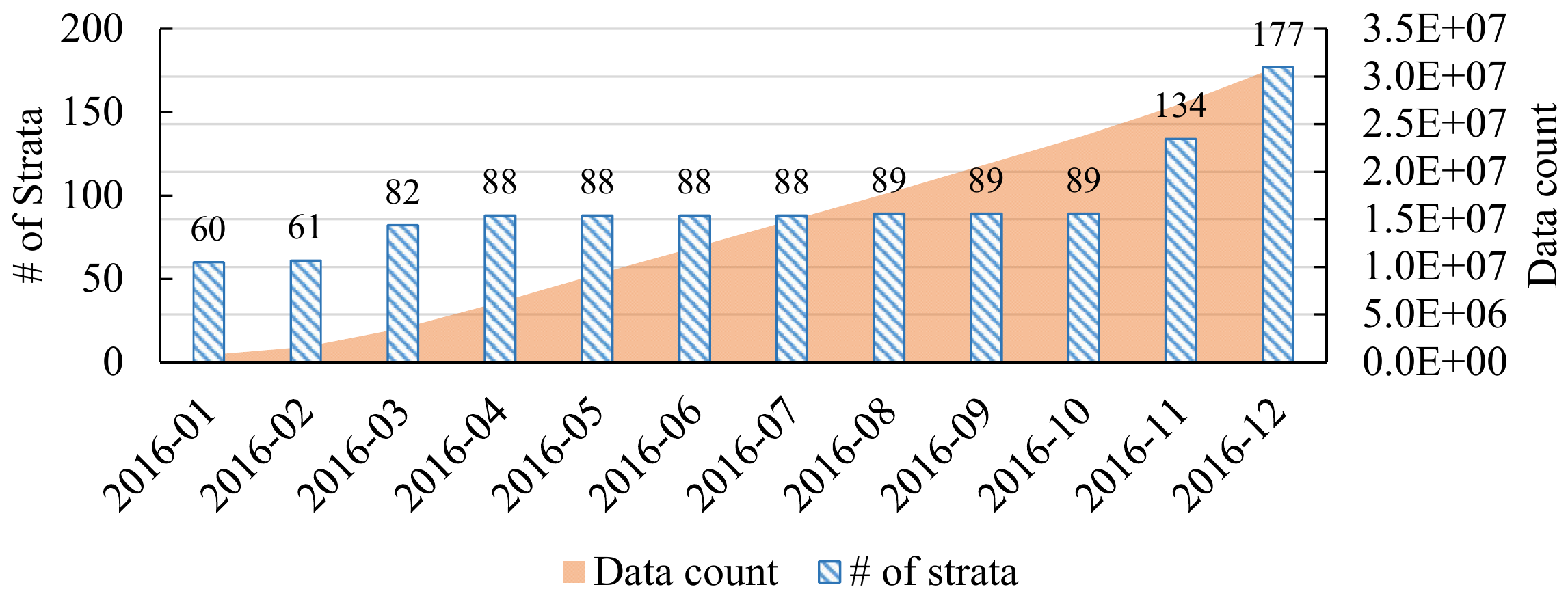}
		\caption{The number of strata seen so far, and the number of records in data, as a function of time. }
		\label{fig:num_stratum}
	\end{subfigure}
	\caption{Characteristics of the OpenAQ dataset.}	
	\label{fig:data_char}
\end{figure}
\begin{figure}[t]
	\centering
	\begin{subfigure}{\graphwidth}
		\includegraphics[width=\textwidth]{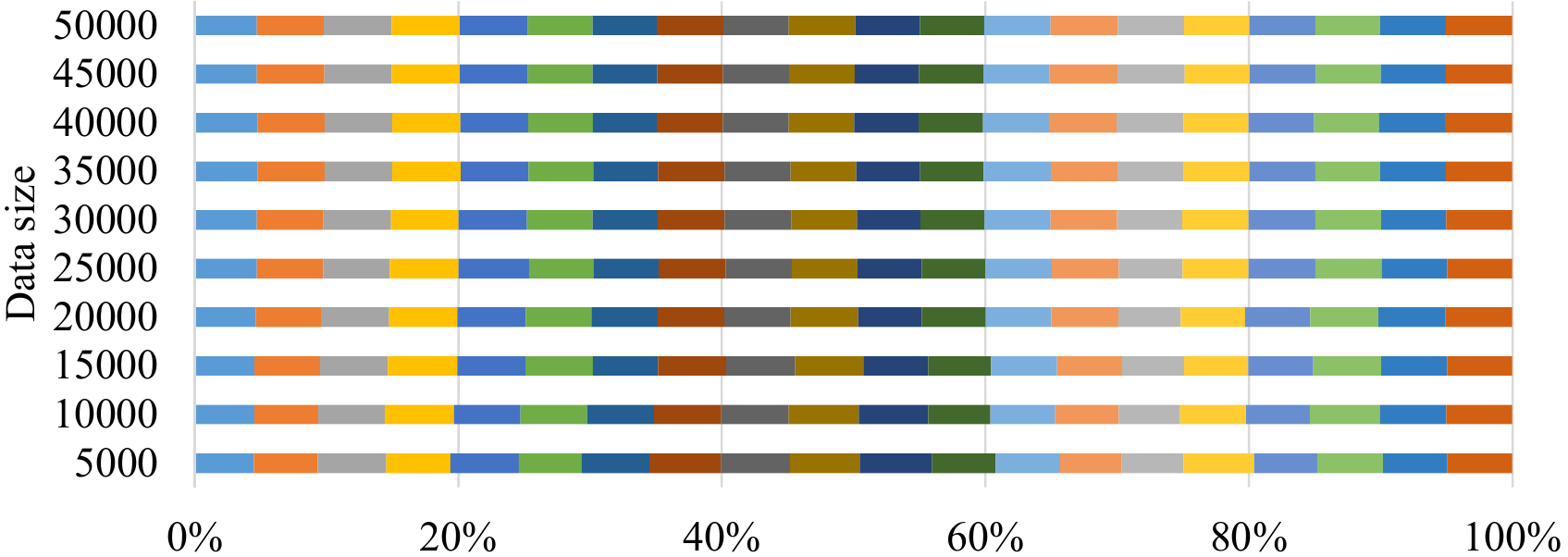}
		\caption{Relative Frequencies of Different Strata.}
		\label{fig:data2_freq}
	\end{subfigure}
	~
	\begin{subfigure}{\graphwidth}
		\includegraphics[width=\textwidth]{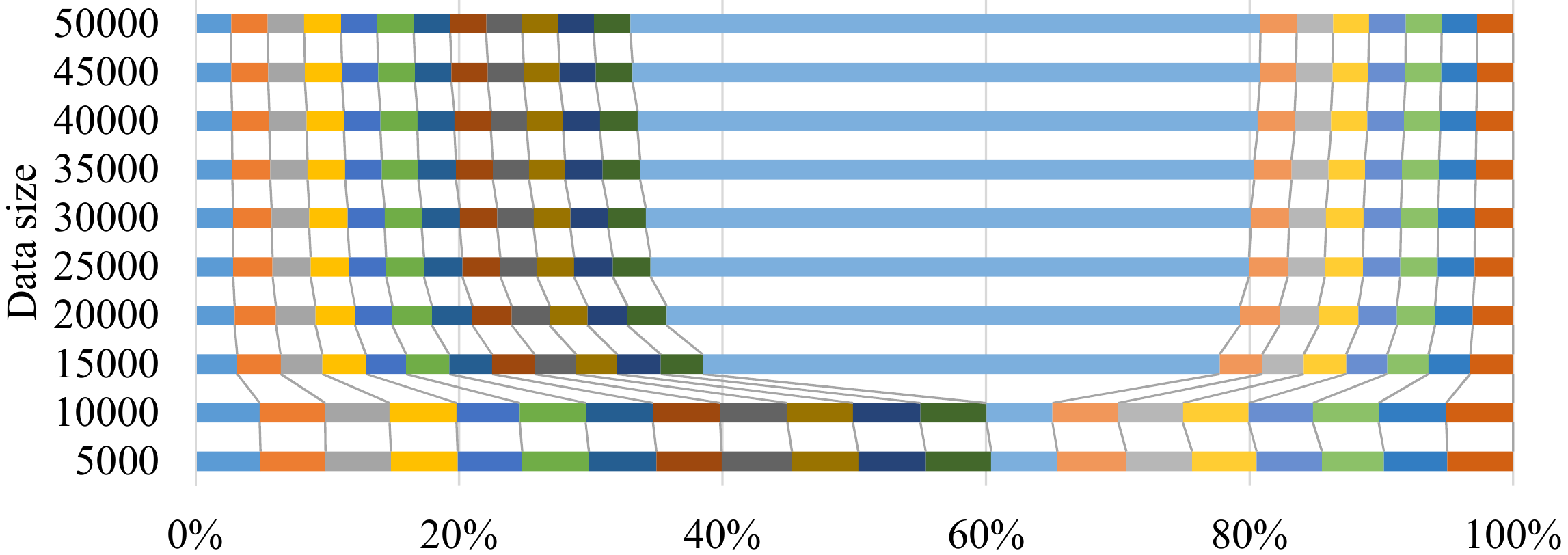}
		\caption{Relative Standard Deviations of Different Strata}
		\label{fig:data2_stddev}
	\end{subfigure}
	\label{fig:data2_char}
	\caption{The Change in Data Characteristics over time, for the Synthetic dataset.}	
\end{figure}

We used two datasets. The first is the OpenAQ dataset~\cite{OpenAQ}, which
contains more than $31$ million records of air quality measurements from
$7,923$ locations in $62$ countries around the world in 2016. The
measurements includes particulate matter (PM10 and PM2.5), sulfur
dioxide (SO$_2$), carbon monoxide (CO), nitrogen dioxide (NO$_2$),
ozone (O$_3$), and black carbon (BC). Data is replayed in time
order to generate the data stream and is stratified based on the
country of origin and the type of measurement, e.g., all measurements of
carbon monoxide in the USA belong to one stratum, all records of
sulphur dioxide in India belong to another stratum, and so on. The
total number of strata at the end of observation is 177, as shown in
Figure~\ref{fig:num_stratum}.

We note that each stratum begins with zero records, and in the initial
stages, each stratum is bounded. As more data are observed,
many of the strata are not bounded anymore, but it is still the case
that there are some strata with few observations, when compared
with other strata. Further, new strata are added 
as more sensors are incorporated into the data stream.
Figure~\ref{fig:num_stratum} shows that new strata are being added with time.
Figure~\ref{fig:data_freq} and \ref{fig:data_stddev} respectively show
the cumulative frequency and standard deviation of the data over time.
As seen, the relative frequency and relative standard deviation of
different strata change significantly. As a result, the
variance-optimal sample-size allocations to strata also change over
time, and the streaming algorithms need to adapt to these changes.

The characteristics of real data, including number and properties of strata are changing frequently and continuously, 
and the allocation is a result of the combined adaptation due to multiple changes.
In order to evaluate on data over which we have more control, 
we created a synthetic data source. Each record $i$ from this source is a tuple
$\left\langle {{sid},{val}} \right\rangle$ where $sid$ is the id of
the stratum the record belongs to, $val$ is the value. The number of strata is set to $20$. Frequencies are equal between
strata, \ie at any time, each stratum has approximately same amount of
records. For a stratum $j$, the value of each record is drawn at
random from Gaussian distribution with two parameters mean $\mu_j = 1$
and standard deviation $\sigma_j$. For the first $10,000$
records, we set $\sigma_j = 1$ for all the strata. After that, we
change the standard deviation of stratum 12 by setting $\sigma_{12} = 20$, while keeping the other strata fixed.
Figures \ref{fig:data2_freq} and \ref{fig:data2_stddev} show the
relative frequencies and standard deviations of the synthetic dataset
over time. While the frequencies are stable, the accumulated standard
deviation shows how stratum 12 changes. 

\subsection{Allocations to Different Strata}
We measured the allocation of samples to different strata. 
Unless otherwise specified, the sample size $M$ is set to 1 million records. 
The allocation can be seen as a vector of~numbers that sum up to $M$
(or equivalently, normalized to sum up to $1$), and we
observe how this vector changes as more elements arrive.
Figure~\ref{fig:alloc_combined} shows the allocations at a single
point in time, at the end of September 2016, for OpenAQ data. From this
figure, we see that the allocation of the streaming sampler \svoila tracks that of the
variance-optimal offline sampler \voila quite closely. As expected,
\reservoir's allocation is proportional to the volume of the stratum,
while \ssunif's allocation is the same across all strata.

\begin{figure}[t]
	\centering
	\includegraphics[width=\graphwidth]{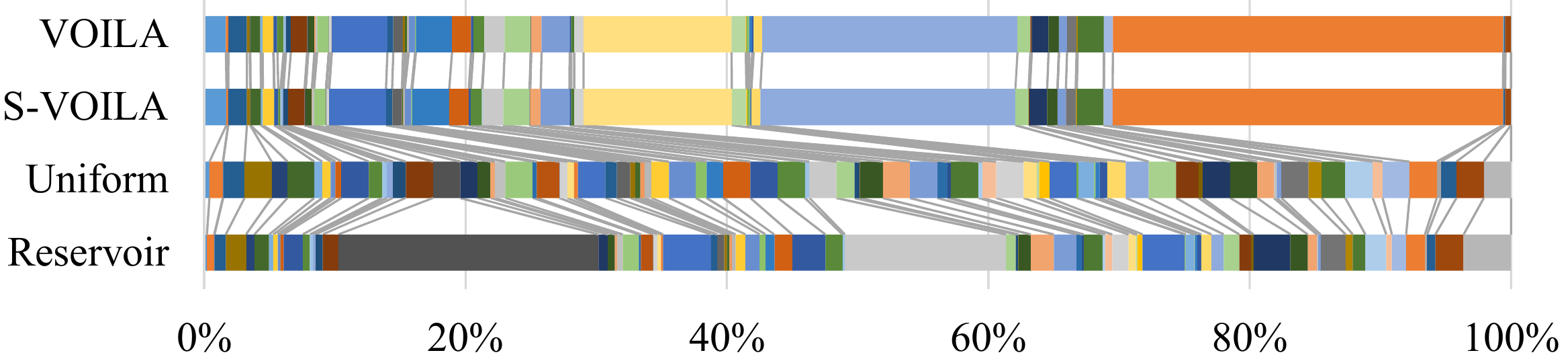}
	\vspace{-0.1in}
	\caption{Allocation of sample sizes among strata after 9 months, OpenAQ data}
	\label{fig:alloc_combined}
\end{figure}

\begin{figure*}[t]
	\centering
	\begin{subfigure}[t]{0.32\textwidth}
		\includegraphics[width=\textwidth]{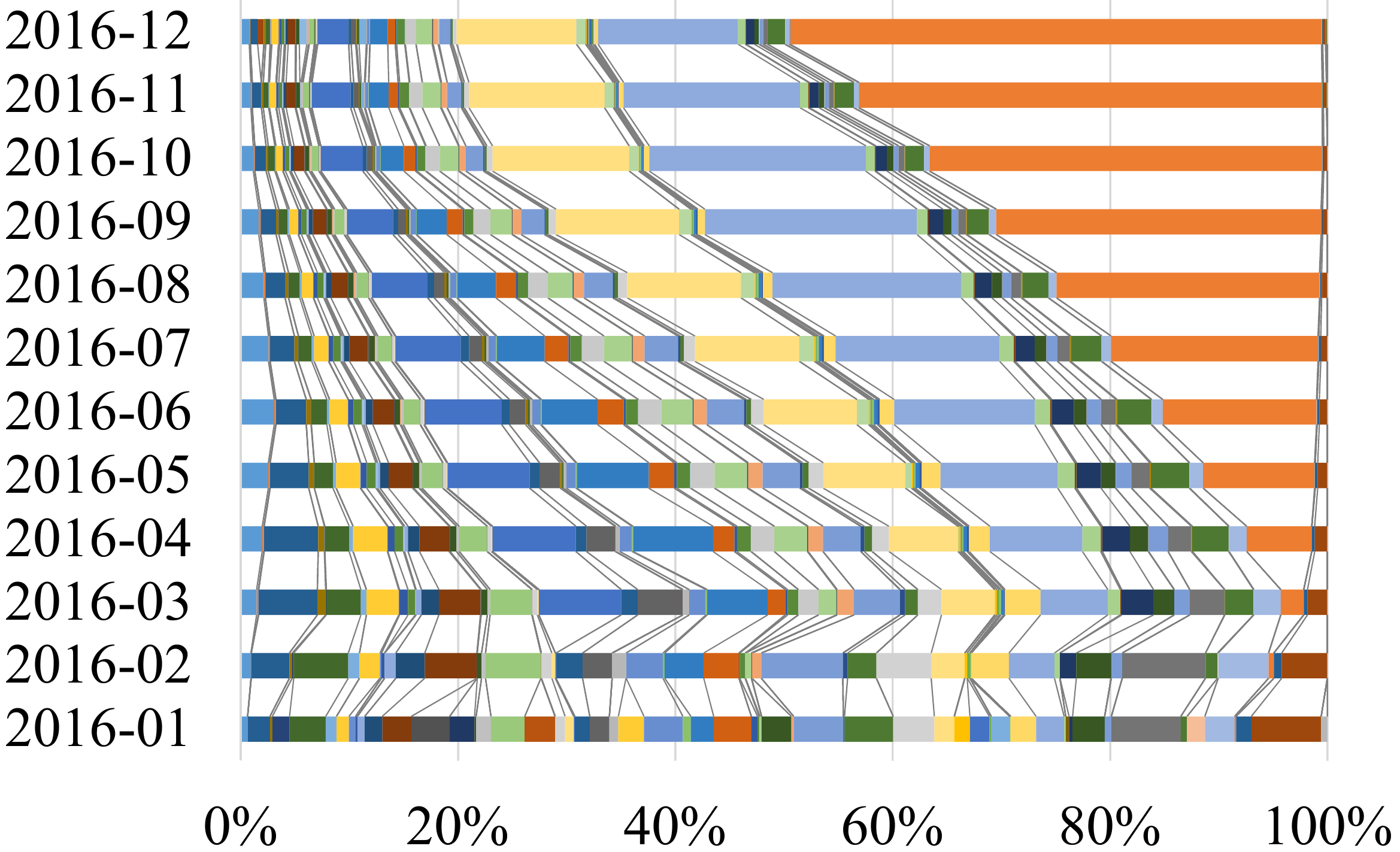}
		\caption{Allocation due to \voila}
	\label{fig:alloc_neyman_offline}
	\end{subfigure}\hfill%
	\begin{subfigure}[t]{0.32\textwidth}
		\includegraphics[width=\textwidth]{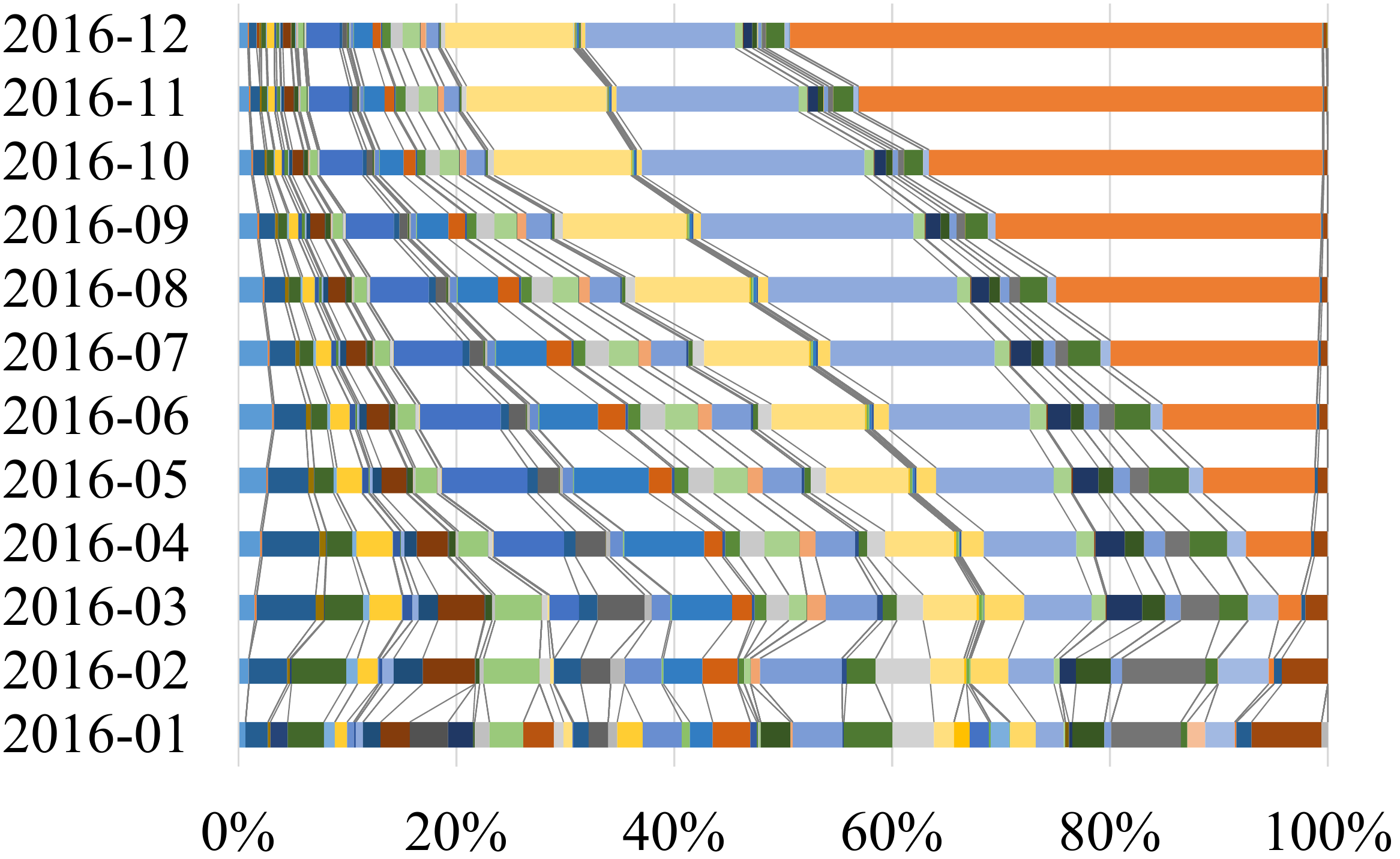}
		\caption{Allocation due to \svoila with Single Element Processing}
		\label{fig:alloc_neyman_online}
	\end{subfigure}\hfill%
	\begin{subfigure}[t]{0.33\textwidth}
		\includegraphics[width=\textwidth]{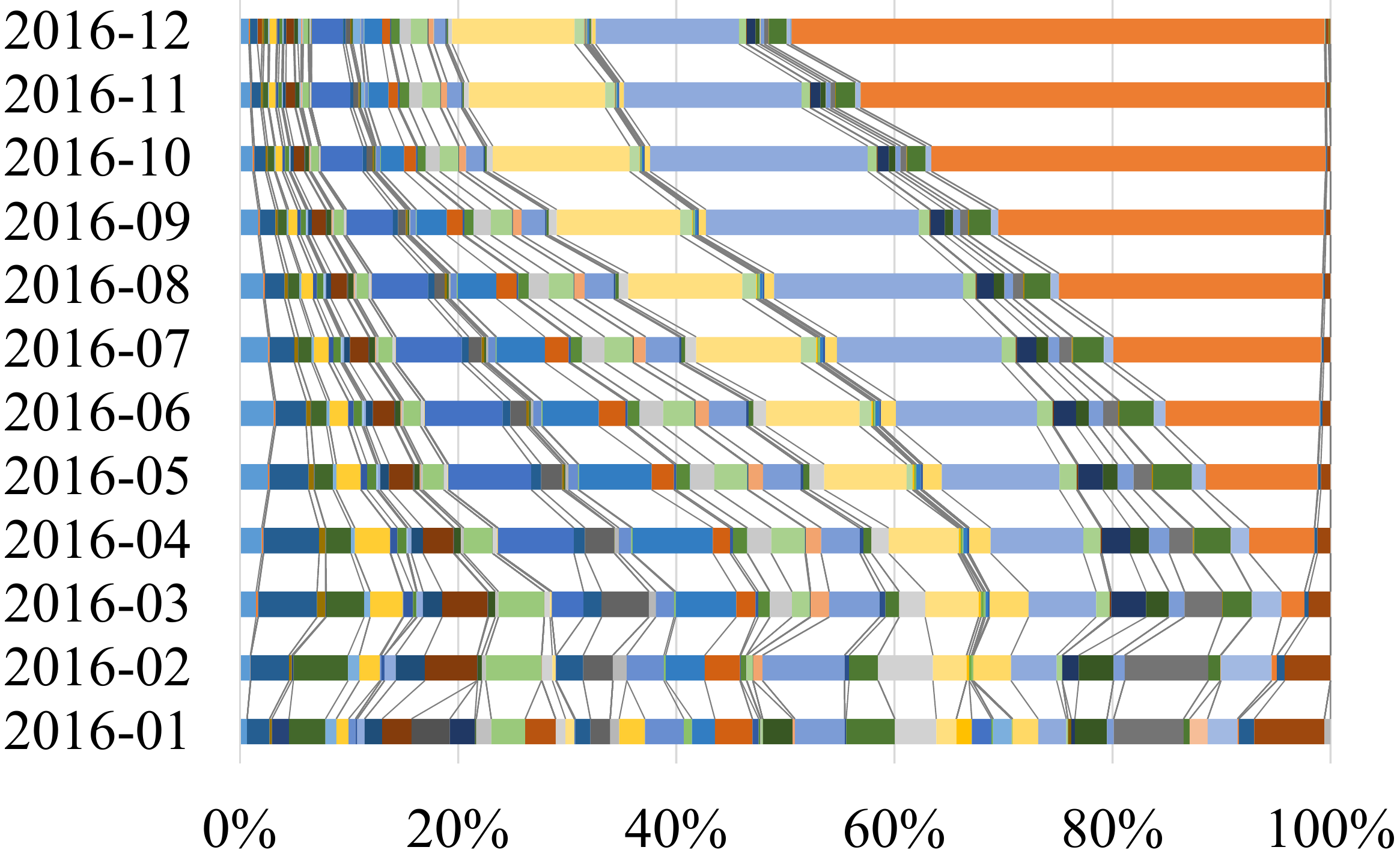}
		\caption{Allocation due to \svoila with MiniBatch
                  Processing (batch size = one-day of data)}
		\label{fig:alloc_neyman_batch}
	\end{subfigure}
	\label{fig:allocation}
	\caption{Change in Allocation over time, OpenAQ data}	
\end{figure*}


\begin{figure}[t]
	\centering
	\includegraphics[width=\graphwidth]{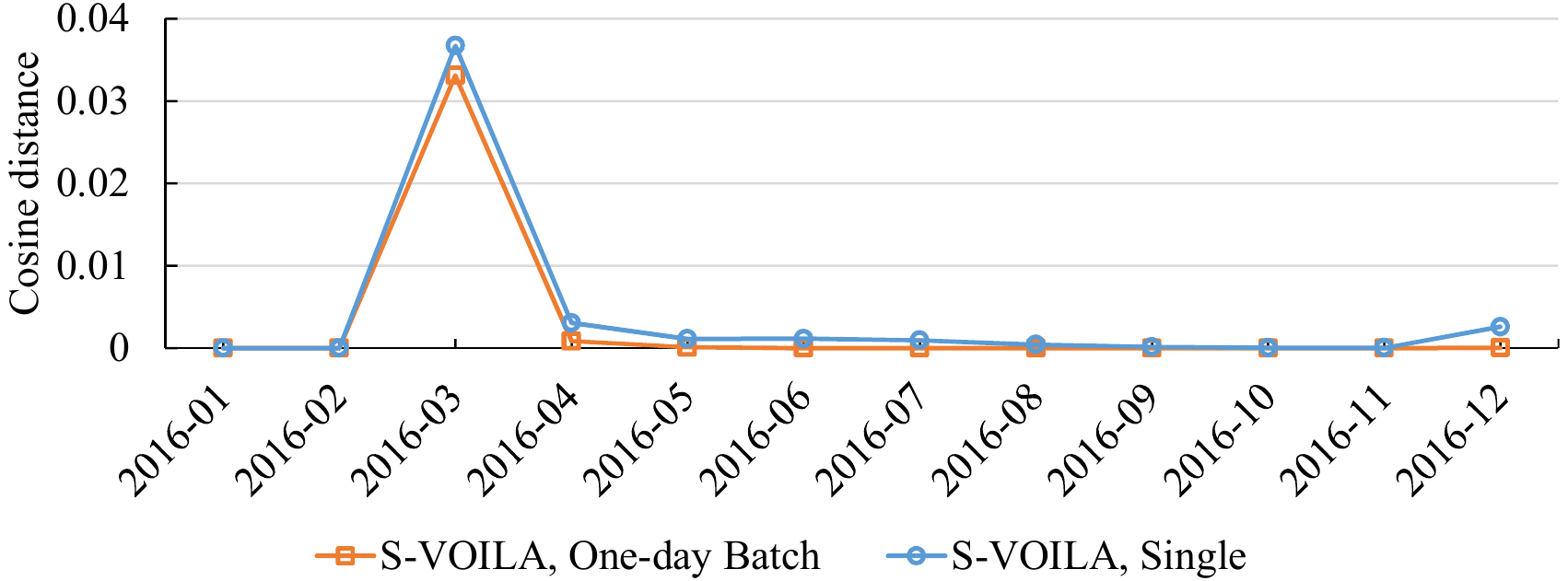}
	\vspace{-0.1in}
	\caption{Cosine distance between the allocations due to \voila,
          \svoila with single element processing, and \svoila with minibatch processing, OpenAQ data.}
	\label{fig:dist_cosine}
\end{figure}


Figures \ref{fig:alloc_neyman_offline}, \ref{fig:alloc_neyman_online}
and \ref{fig:alloc_neyman_batch} show the change in allocations over time
resulting from \voila, \svoila with single element processing, and
\svoila with minibatch processing (minibatch size = 1 day's data). 
Visually, the allocations produced by the three
methods track each other over time, showing that the streaming methods
follow the allocation of the optimal offline algorithm, \voila. To
understand the difference between the allocations due to \voila and
\svoila quantitatively, we measure the cosine distance between the
allocation vectors from \voila and \svoila. The results show that allocation vectors due to \svoila and \voila
are very similar, since the cosine distance is close to $0$ most of the
time and less than $0.04$ at all times. We further note that \svoila with minibatch
processing yields an allocation that is closer to \voila than \svoila
with single element processing.


\begin{figure}[t]
  \centering
  \includegraphics[width=\graphwidth]{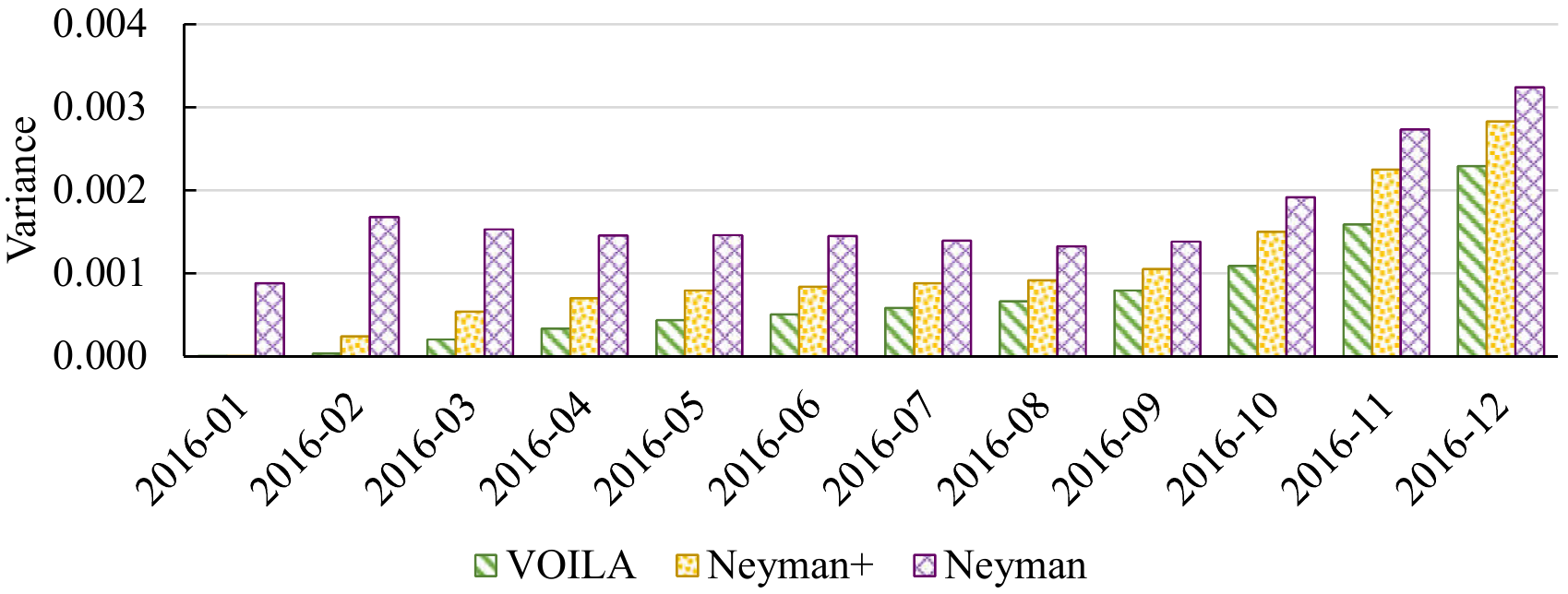}
  \caption{Variance of \voila compares to \neyman and \neymanPlus with equal sample size:f 1M records, OpenAQ data.} \label{fig:alloc_variance_2} 
\end{figure}
\begin{figure}[t]
  \centering
  \includegraphics[width=\graphwidth]{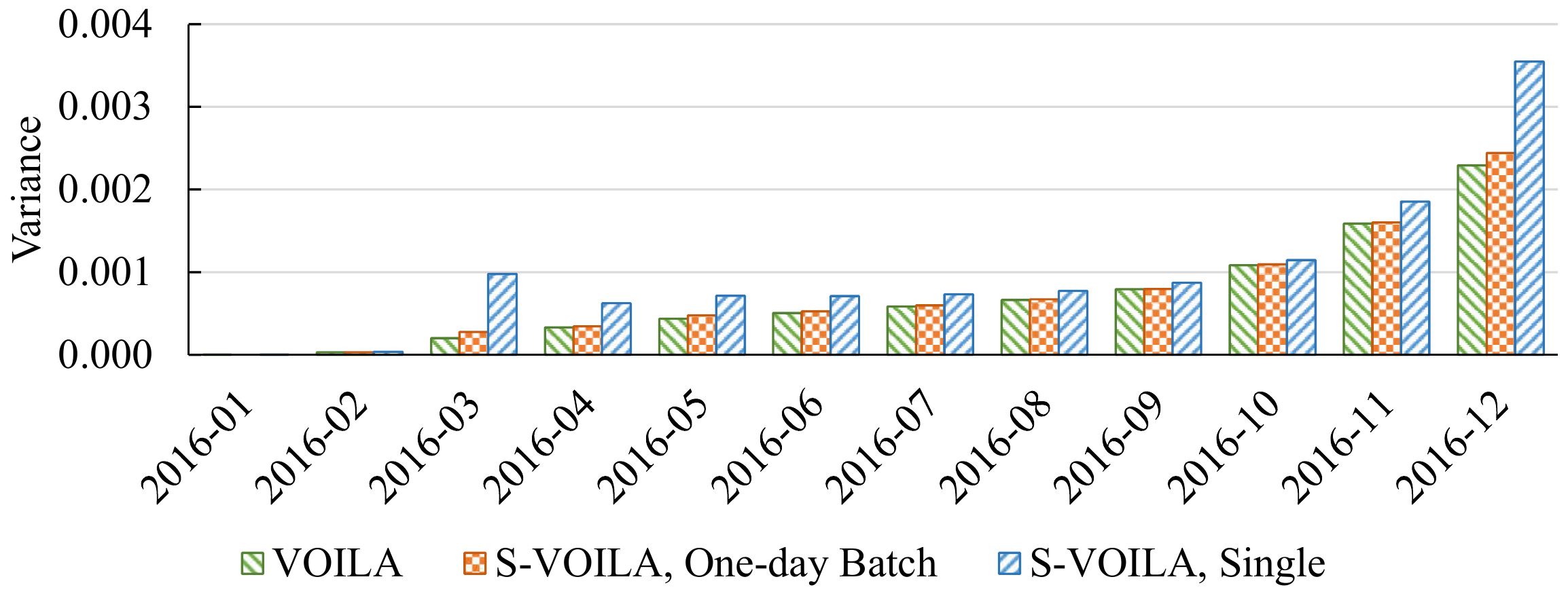}
  \caption{Variance of streaming \svoila, with Single and
    Minibatch Processing, compared with offline \voila. Sample size
    is set to 1M records, for each method, OpenAQ data.}
  \label{fig:alloc_variance}
\end{figure}

\begin{figure}[t]
  \centering
  \includegraphics[width=\graphwidth]{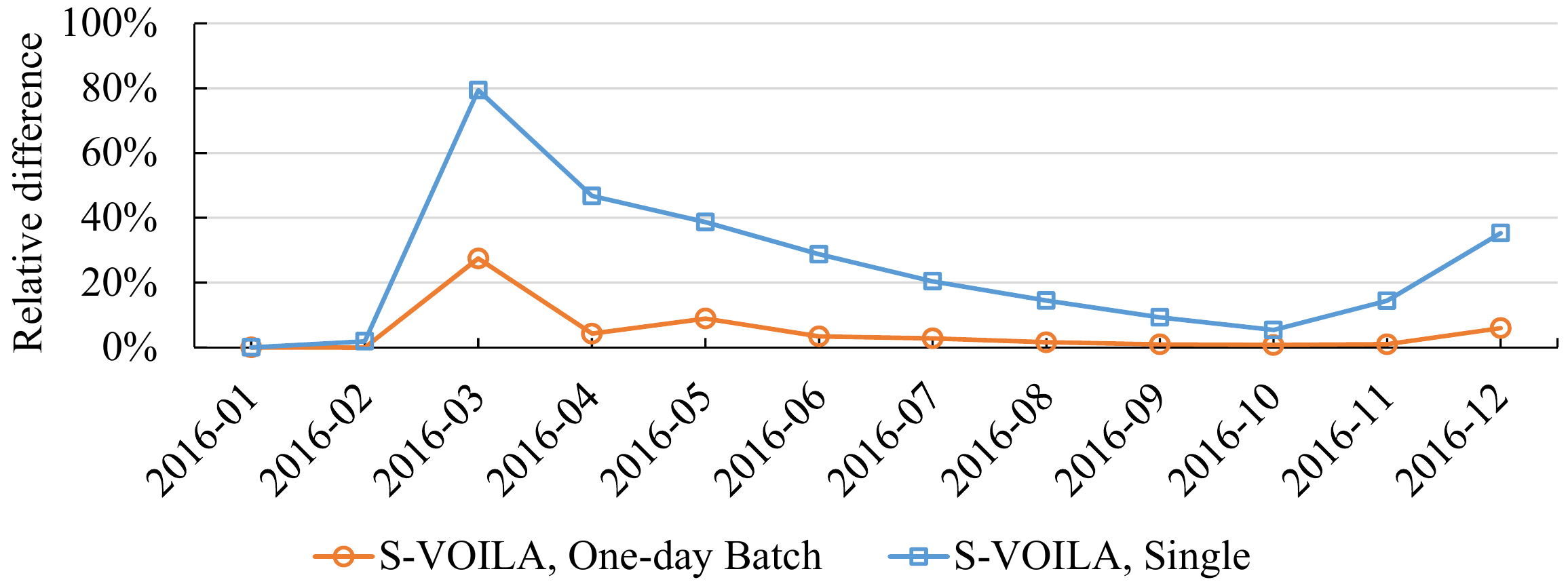}
  \caption{Relative difference of the variance of \svoila, with
    Single and Minibatch Processing, compared with the optimal
    variance due to \voila, OpenAQ data.}
  \label{fig:dist_variance}
\end{figure}

\subsection{Comparison of Variance} 
We compared the variance of the estimates (Equation~\ref{eqn:var}) from the stratified
random samples produced by different methods, offline or streaming. 
The results are shown in Figures~\ref{fig:alloc_variance_2},~\ref{fig:alloc_variance} and~\ref{fig:dist_variance}.
Generally, the variance of the sample due to each method increases over time, since the volume of data as well as the
number of strata increase, while the sample size is fixed.

Among offline algorithms, we observe from Figure~\ref{fig:alloc_variance_2} 
that \neyman results in a variance that is larger than \voila, by a factor of 1.4x to 50x. While \neyman is known to be
variance-optimal for unbounded strata, these results show that it is far from variance-optimal for bounded strata. 
\voila is better than \neyman in two respects: (1)~it uses all available memory, and (2)~it allocates 
memory among strata in an optimal fashion. 
In order to measure the impact of the allocation, we compared the variance of
\voila with that of \neymanPlus, which uses all available memory.
From Figure~\ref{fig:alloc_variance_2} we observe the following. First, \voila always has a lower variance than
\neymanPlus and \neyman -- note this also implies that at point in time, there are bounded strata in the OpenAQ data, since 
otherwise, \neyman would also result in optimal variance. Second, the variance due to \voila is always smaller than
the variance due to \neymanPlus, by a factor of 1.2x to 7.1x. This shows that carefully dealing with bounded
strata using \voila can lead to significantly better stratified random samples.



The comparison of the variance of streaming algorithms is shown in Figure~\ref{fig:alloc_variance_2}.
Among the streaming algorithms, we note that the variance due to \svoila with single element
processing and with minibatch processing are typically close to that of the optimal algorithm, \voila.
The variance  of \svoila using minibatch processing is very close to that of \voila, showing that it is nearly variance-optimal at
all times. The variance of \svoila with single element processing is typically worse than minibatch processing.

Figure~\ref{fig:dist_variance} shows the relative difference between
the variance produced by a streaming algorithm ($\hat{x}$) and the
optimal variance due to \voila ($x$), defined as
$\frac{\hat{x} - x}{x}$. We note that the variance of both variants of \svoila are nearly equal to that
of \voila until March, when they start increasing relative to \voila, and then converge back.

From analyzing the underlying data stream, we see that March is the
time when a number of new strata appear in the data (Figure~\ref{fig:num_stratum}), causing
significant changes in the optimal allocation of samples to strata
(this can also be seen in Figure~\ref{fig:dist_cosine} showing the
cosine distance between the allocations). An offline algorithm such as
\voila can resample more elements from a stratum, if necessary, since
it has access to all data from the stratum. However, a streaming
algorithm such as \svoila cannot do so and must wait for enough new
elements to arrive in these strata before it can ``catch up'' to the
allocation of \voila. Hence, \svoila with single element as well as
with minibatch processing start showing an increase in the variance at
such a point. When data becomes stable again, and more elements arrive,
the relative performance of \svoila improves. \svoila with minibatch
processing approaches the optimal variance faster than \svoila
with single element processing, which is as expected, since as the
size of the minibatch increases, better optimization decisions are
made with respect to which elements to exclude from the sample. In
November and December, new strata appear again, and the relative
performance is again affected. Overall, we note that \svoila with
minibatch processing produces variance that is significantly closer to
\voila than \svoila with single element processing.

\begin{figure}[t]
	\centering
	\includegraphics[width=\graphwidth]{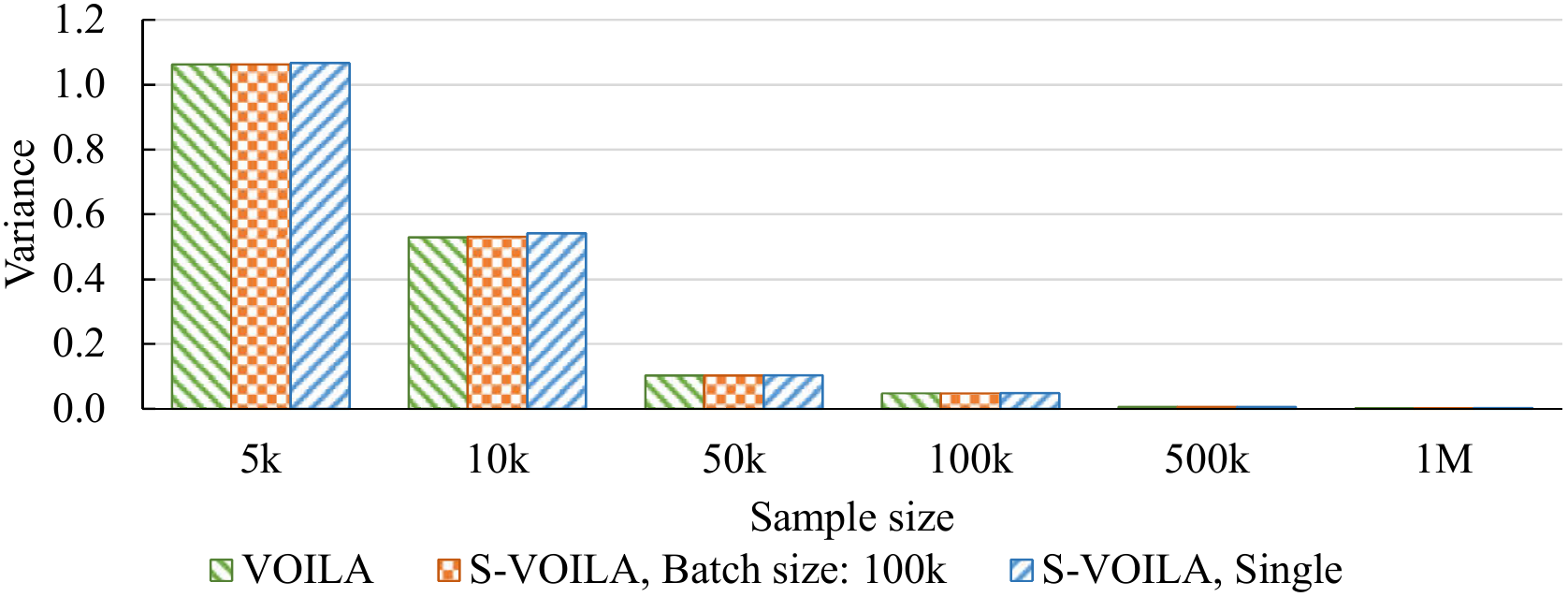}
	\caption{Impact of Sample Size on Variance, in September, OpenAQ data.}
	\label{fig:alloc_variance_sample_size}
\end{figure}

\begin{figure}[t]
	\centering
	\includegraphics[width=\graphwidth]{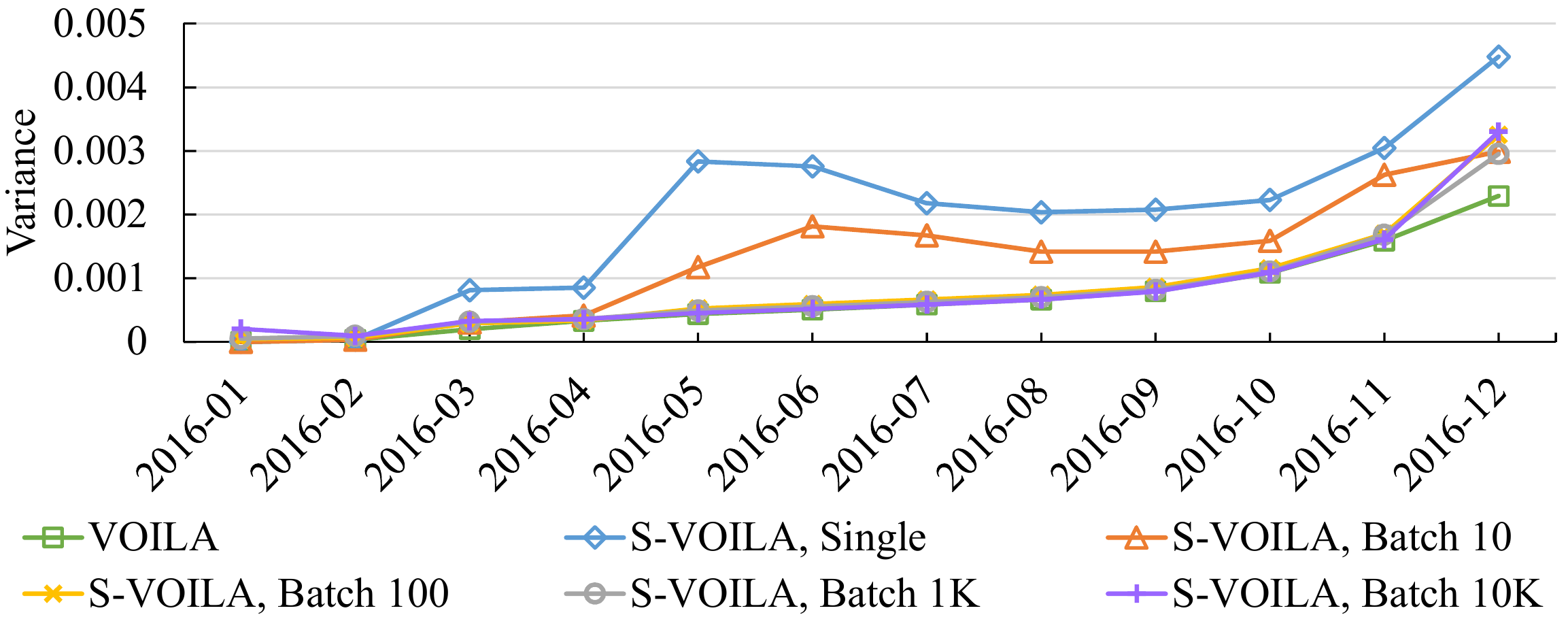}
	\caption{Impact of Batch Size on Variance, OpenAQ data.}
	\label{fig:alloc_variance_batch_size}
\end{figure}

{\bf Impact of Sample Size:} To understand the sensitivity to the size of the sample, we conducted an experiment where the sample size is varied from 5000 to 1 million records. We fixed the batch size to 100 thousand records. Figure \ref{fig:alloc_variance_sample_size} shows the snapshot in September 2016 of variances as a function of the sample size. Both \voila and \svoila, with single~element and minibatch processing, the variance decreases when the sample size increases. This is as expected, since larger samples produces better estimates of the population mean.

{\bf Impact of Batch Size:} It is clear from Figure~\ref{fig:alloc_variance} that the variance of minibatch \svoila, where each batch contains data collected in a day, is significantly smaller than that of single element \svoila. In order to better understand the impact of the batch size, we conducted an experiment where we tried different batch sizes for minibatch streaming \svoila, chosen from $\{1,10,{10}^2, {10}^3, {10}^4\}$. The results are shown in Figure~\ref{fig:alloc_variance_batch_size}. A batch size of 10 elements yields significantly better results than single element \svoila. A batch size of 100 or greater makes the variance of \svoila nearly equal to the optimal variance.

\subsection{Query Performance}
We now evaluate the quality of these samples indirectly, through their
use in approximate query processing, which is one of the major
applications of sampling. The streaming sampler continuously maintains a stratified random
sample of data (stored in memory), and use this sample to
approximately answer aggregate queries, which are issued by the
client. The offline sampler constructs its sample when needed, using
\voila, which takes two passes through the data. For evaluating the
approximation error in query processing, we also implement an exact
method for query processing, $\exact$, that stores every record in a
table (stored in a MySQL database~\cite{mySql}) and answers a query
using this table. While the exact method has zero error, its
processing time is high, and so is its space overhead. Identical
queries are made at the same time points in the stream to the
different streaming and offline samplers, as well as to the exact
query processor.

\begin{figure}[t]
	\centering
	\includegraphics[width=\graphwidth]{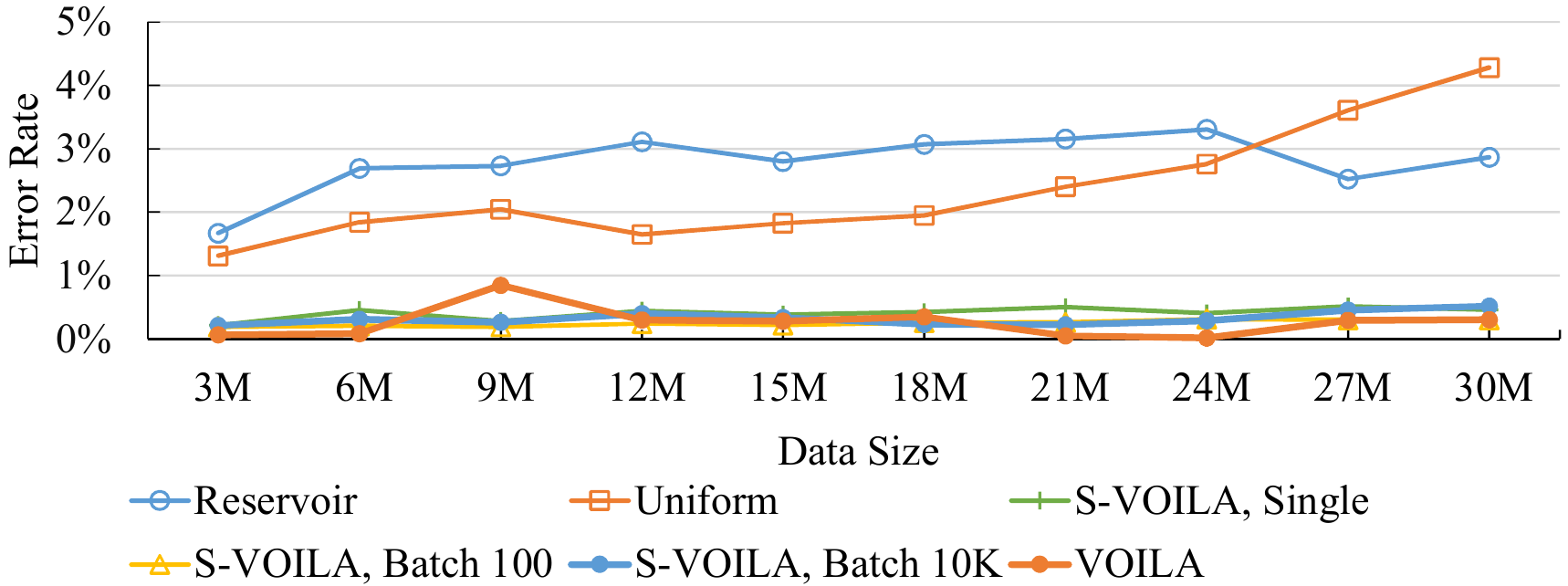}
	\caption{Query Performance as data size varies, with sample size fixed at 100,000. OpenAQ data.}
	\label{fig:dataSize}
\end{figure}

\begin{figure}[t]
	\centering
	\includegraphics[width=\graphwidth]{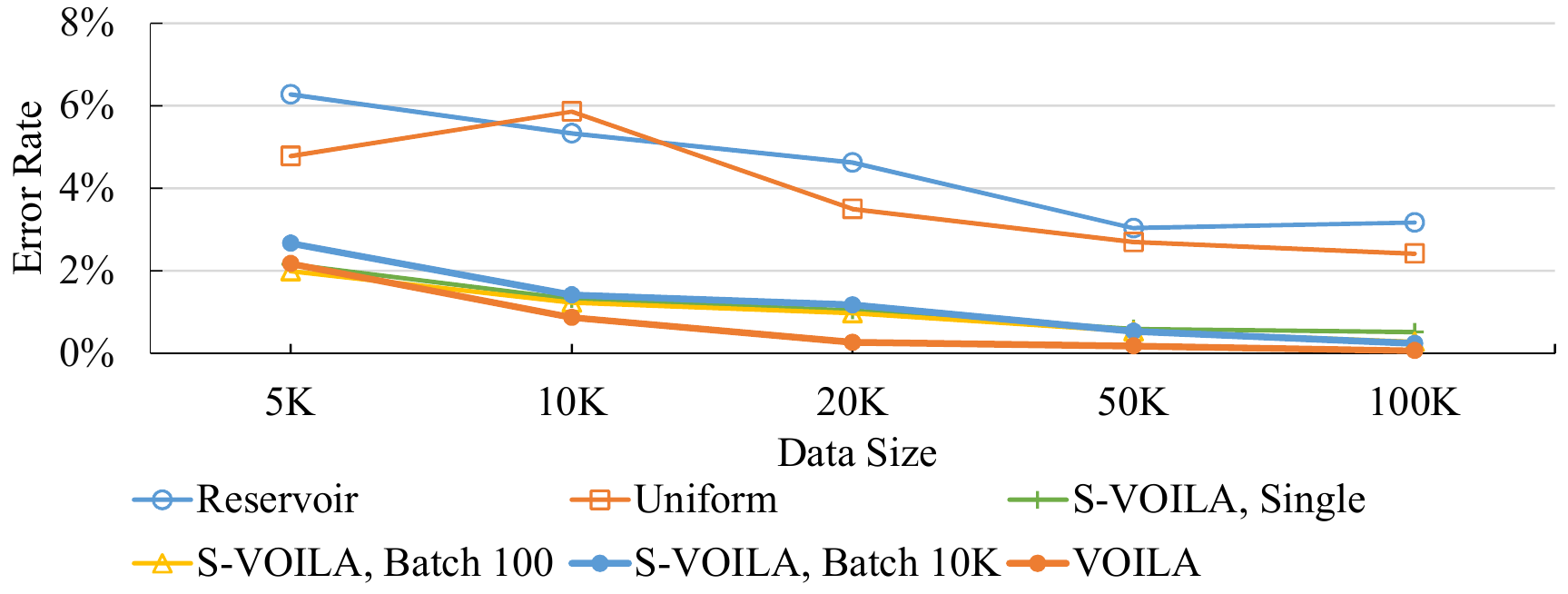}
	\caption{Query Performance as sample size varies, with data size fixed at 21 million. OpenAQ data.}
	\label{fig:sampleSize}
\end{figure}


We measure the accuracy in query processing of the following samplers: \reservoir,
\ssunif, \svoila, \vob, and \exact. We use the metric of relative
error between the approximate answer and the exact answer, where the
query asks for the mean of the data received across all strata. The
sample size is set to $100,000$ for all samplers. For \svoila, we set
minibatch size to be $1$, $100$, and $10,000$. Each data point is the
mean of nine repetitions of the experiment with the same
configuration.

Figure~\ref{fig:dataSize} shows the relative error as the size of the
streaming data increases, while the sample size is held fixed. The
query was executed every three million element arrivals, up to thirty
million, which covers the entire year of 2016 in the OpenAQ dataset.
We note that the relative performance between different methods
remains similar for most data sizes. \reservoir has a consistent
errors since it is mainly affected by sample size rather than data
size. \ssunif is affected by total number of strata and as expected,
we see an increasing error when the data size reaches 24 million,
where the number of strata increases suddenly as shown in
Figure~\ref{fig:num_stratum}, November 2016. The performance of \vob
and \svoila increase slightly with data size, though at much lower
rates than \reservoir and \ssunif. We note that \svoila with any
minibatch size is very close to \vob.

Figure~\ref{fig:sampleSize} shows the impact of the sample size, as it
varies from 5,000 to 100,000, and the queries were executed at a fixed
time of stream to see how sample size would affect the accuracy of
answering queries. As expected, all methods benefit from increased
sample size. We observed \svoila and \vob perform significantly better
than \reservoir and \ssunif even with smaller sample sizes. Another
observation of \svoila is that a larger minibatch size does not always
guarantee better accuracy. When total sample size is small, each
stratum is allocated with a smaller space and there are fewer bounded
strata. Therefore, the eviction made by single and minibatch
processing affected the performance less. With our configuration,
\svoila with minibatch ten thousand elements did not yield a better
accuracy until sample size was set to one hundred thousand.

\subsection{Adapting to a Change in Data Distribution}
\label{sec:syn}


\begin{figure*}[th]
	\centering
	\begin{subfigure}{0.33\textwidth}
		\includegraphics[width=\textwidth]{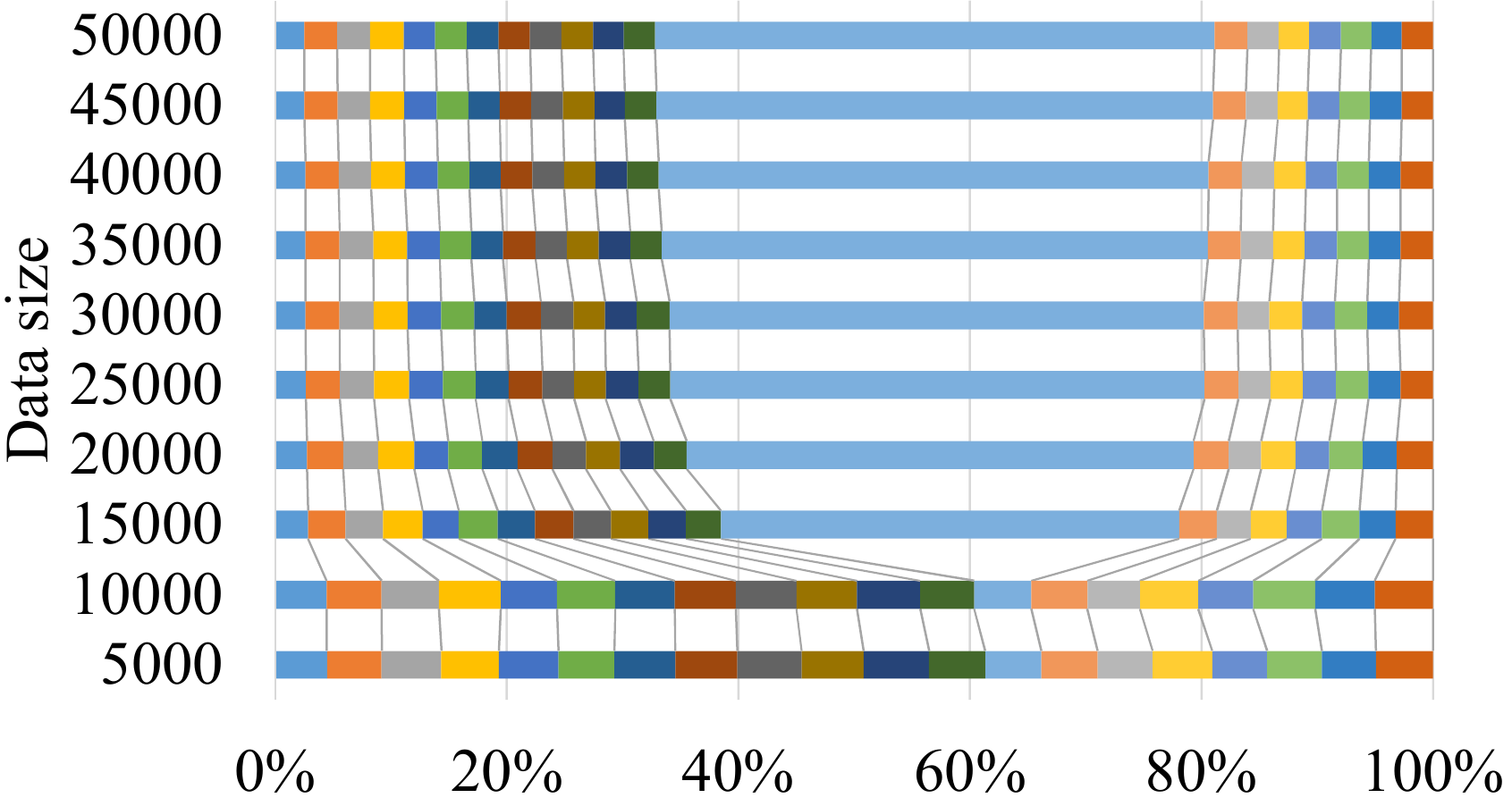}
		\caption{Allocation due to \vob~ across different strata.}
		\label{fig:2alloc_neyman_offline}
	\end{subfigure}\hfill%
	\begin{subfigure}{0.33\textwidth}
		\includegraphics[width=\textwidth]{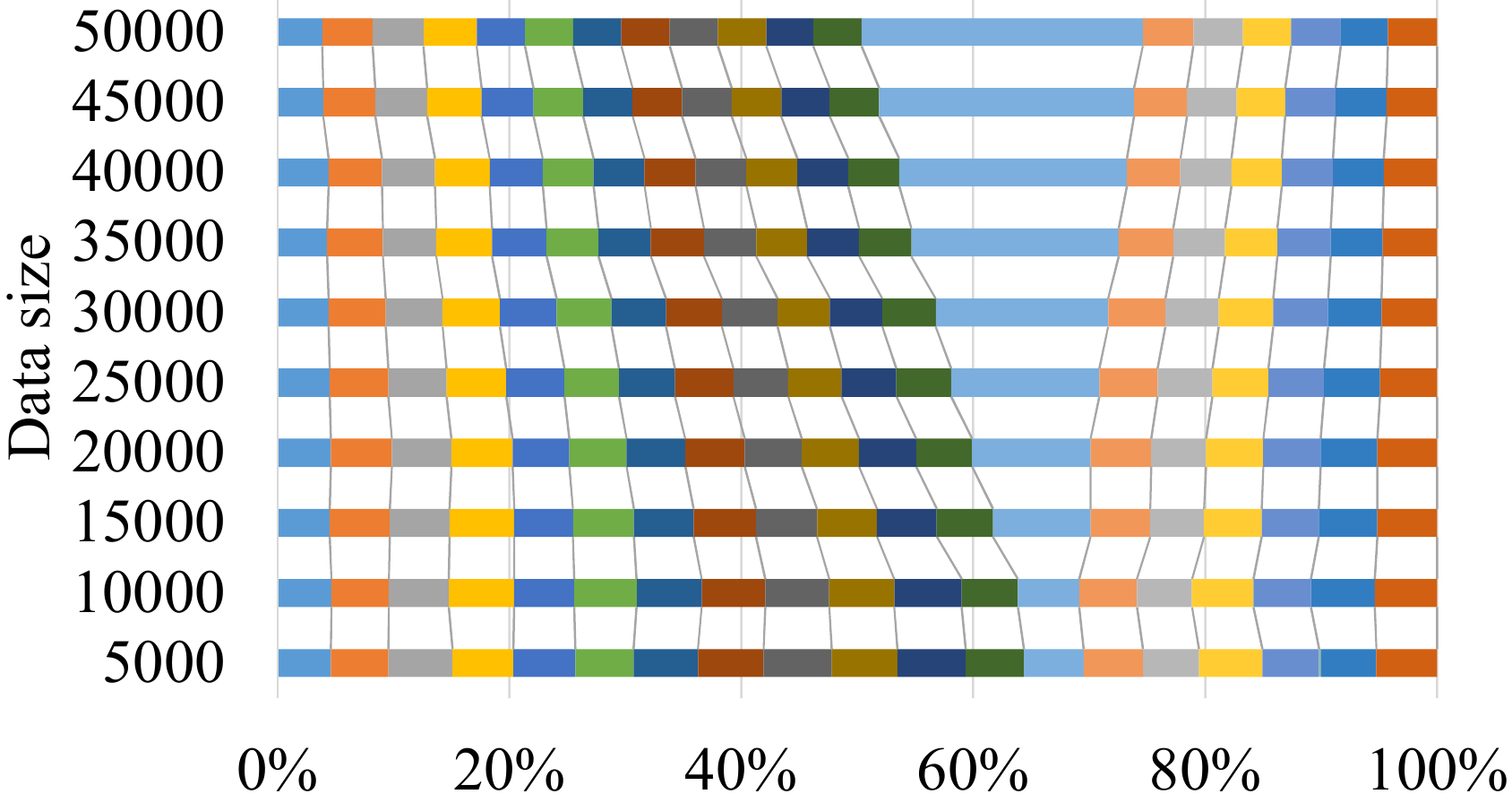}
		\caption{Allocation due to \svoila with Single Element processing.}
		\label{fig:2alloc_neyman_online}
	\end{subfigure}\hfill%
	\begin{subfigure}{0.33\textwidth}
		\includegraphics[width=\textwidth]{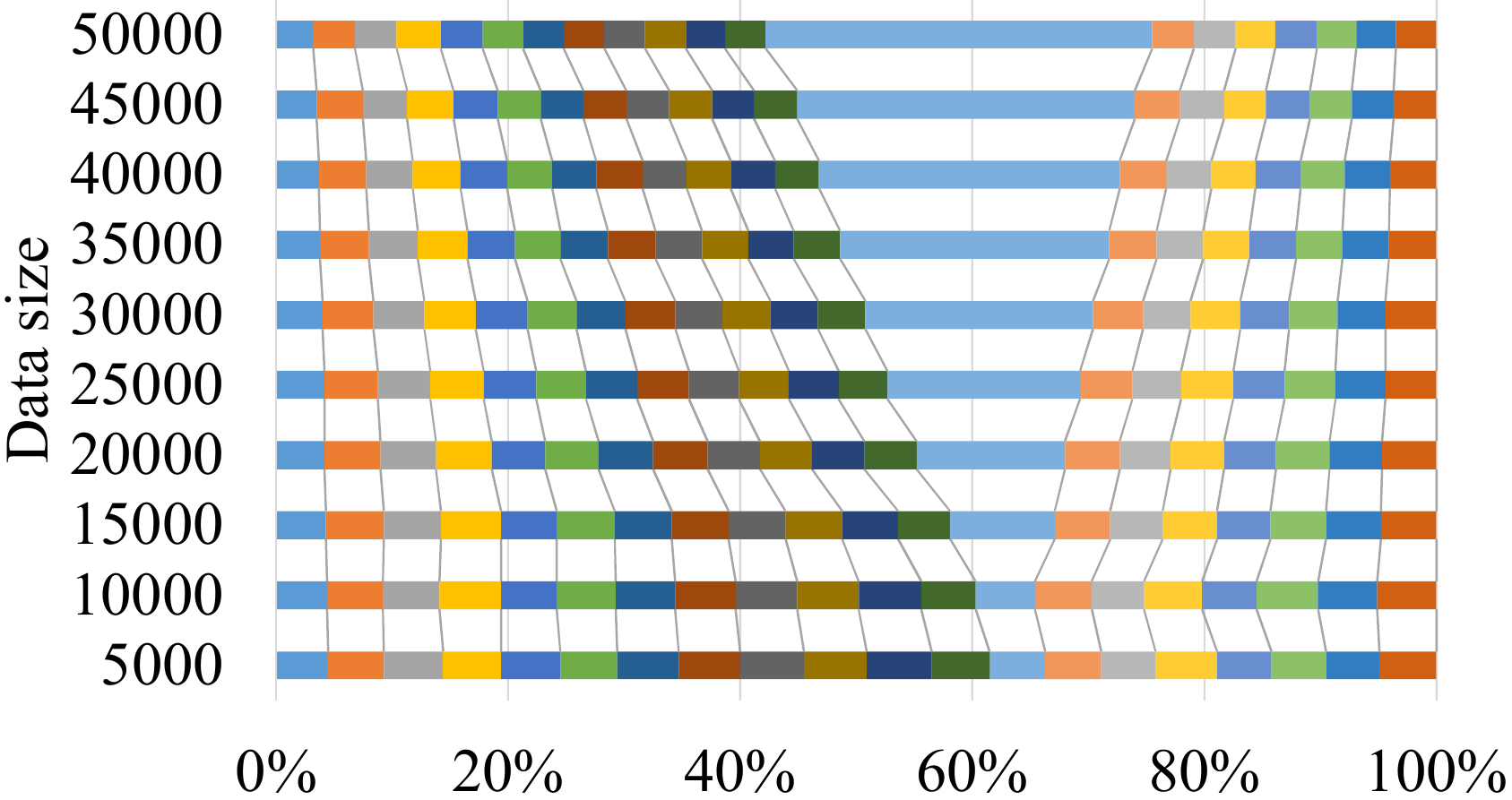}
		\caption{Allocation due to \svoila with Minibatch
                  processing (batch size = 100).}
		\label{fig:2alloc_neyman_batch}
	\end{subfigure}
	\label{fig:2allocation}
	\caption{The Change in allocations of different algorithms over time with synthetic dataset.}	
\end{figure*}

\begin{figure}[th]
	\centering
	\includegraphics[width=\graphwidth]{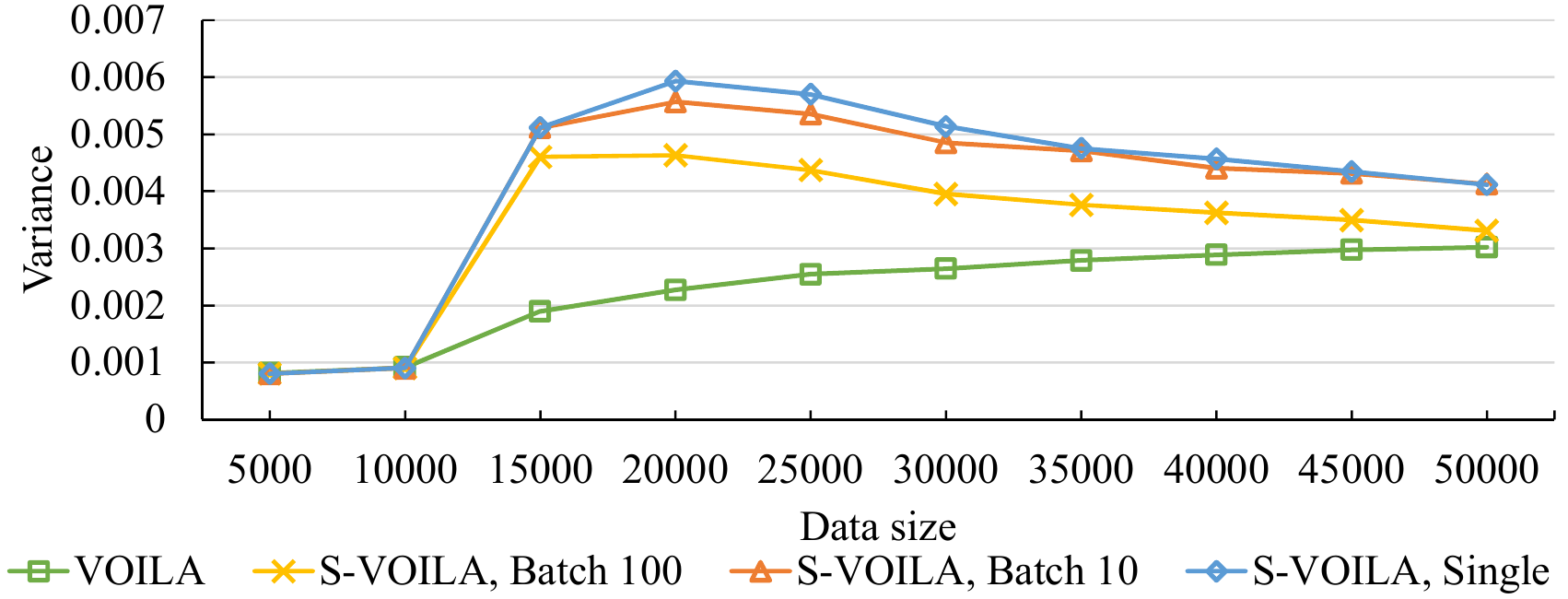}
	\caption{The variance changes due to a sole change in synthetic data.}
	\label{fig:2alloc_variance}
\end{figure}

In a real-world dataset such as OpenAQ, the allocation is affected by
the combination of multiple factors that continuously change. To
better observe the behavior of our algorithms under a single change,
we conducted an experiment with our synthetic data. Figure~\ref{fig:data2_stddev} shows a single change
in stratum 12, where the standard deviation suddenly increases from 1
to 20 after the first $10,000$ records are generated. Meanwhile, the standard deviation of all the other strata are stable and their
frequencies are stable. After this change, we will expect
Stratum 12 to be given a greater sample size than the other strata.
The memory budget is set to $1,000$ records, which is $2\%$ of the data
size at the end of the experiment.

Figures \ref{fig:2alloc_neyman_offline}, \ref{fig:2alloc_neyman_online}, and \ref{fig:2alloc_neyman_batch} show
the allocations produced by \vob, single element \svoila, and
minibatch \svoila, respectively. As seen, \svoila slowly captures the
sudden change in the data by giving Stratum~12 more sample space over time.
\voila is more sensitive to the change, due to the fact that \vob
works in an offline manner and is able to sample more data into
Stratum 12 right after the change. Visually, minibatch \svoila is
closer to the \vob than single element \svoila.

Figure~\ref{fig:2alloc_variance} shows the variance of different methods on synthetic data. At first, when the data is stable, all methods have nearly optimal variance. After a single change at 10,000 records, the variance of \voila increases, while those of different versions of \svoila increase at a faster rate. \svoila with a higher  minibatch size has a lower variance. Interestingly, the variance of all versions of \svoila converge to that of the optimal method, \voila, though \svoila with a minibatch of 100 elements converges the fastest.

\begin{figure}[t]
	\centering
	\includegraphics[width=\graphwidth]{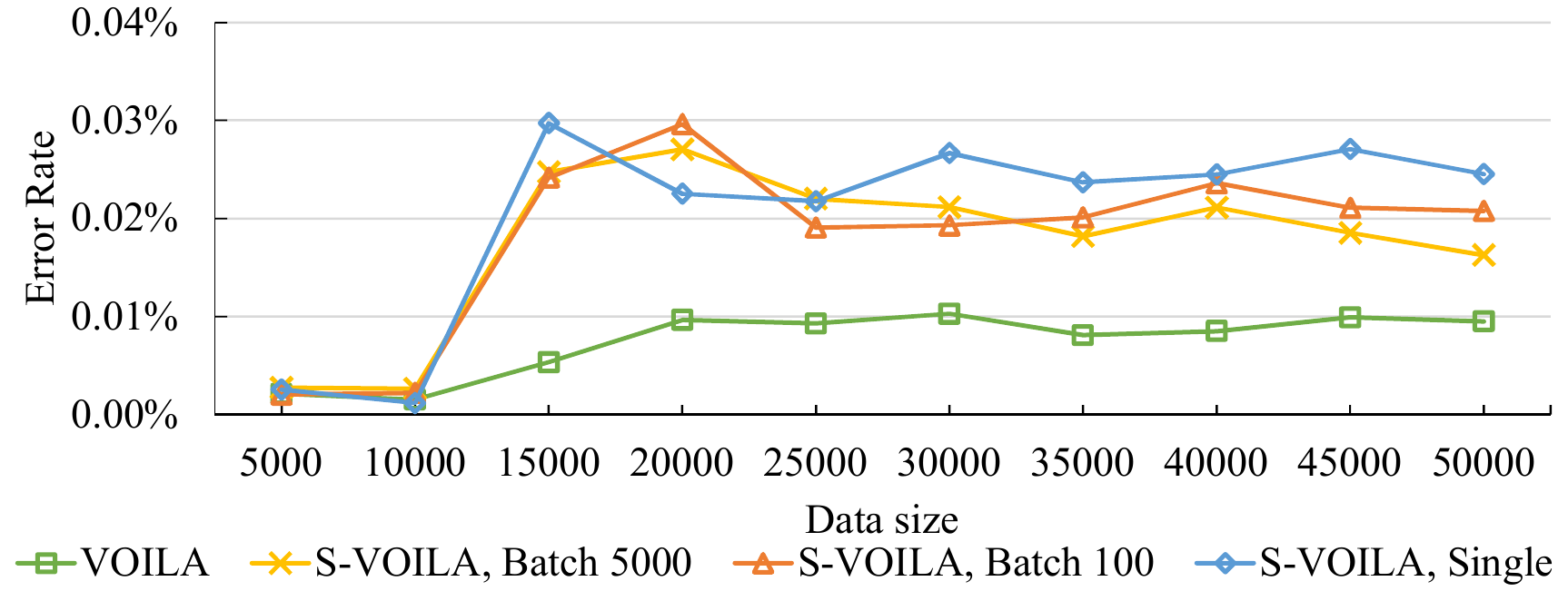}
	\caption{Query Performance on synthetic data as size of streaming data increases, with sample size fixed at 1,000 and one stratum's distribution changed at $10,000$.}
	\label{fig:syn_all}
\end{figure}

We also test the query performance of \vob and \svoila with different minibatch sizes on synthetic data.
Figure~\ref{fig:syn_all} shows the performance of a query across all
strata. The first observation is \vob is less affected by the
distribution change since it samples from all the received data, while
\svoila methods had to discard data on the fly. Another observation is
that performance of \svoila with a larger minibatch size will be closer to \vob as stream continues.


%% file: conclusion.tex
\section{Conclusions}
\label{sec:conclusion}
We presented \voila, a variance-optimal method for offline SRS from data that may have bounded strata. \voila is a generalization of Neyman allocation, which assumes that each stratum has abundant data available. Our experiments show that on real and synthetic data, a stratified random sample obtained using \voila can have a significantly smaller variance than one obtained by Neyman allocation. We also presented \svoila, an algorithm for streaming SRS with minibatch processing, whose sample allocation is continuously adjusted in a locally variance-optimal manner. Our experiments show that \svoila results in variance that is typically close to \voila, which was given the entire input beforehand. The quality of the sample maintained by \svoila improves as the size of the minibatch increases. We show an inherent lower bound on the worst-case variance of any streaming  algorithm for SRS -- this limitation is not due to the inability to compute the optimal sample allocation in a streaming manner, but is instead due to the inability to increase sample sizes in a streaming manner, while maintaining uniformly weighted sampling within a stratum. There are several directions for future research, including (1)~restratification in a streaming manner (2)~incorporating time-decay into sampling, where more recent elements are given a higher probability of being included in the sample, and (3)~stratified random sampling on distributed data.